\documentclass[journal]{IEEEtran}
\IEEEoverridecommandlockouts

\newif\ifarxiv\arxivtrue

\usepackage{cite}

\ifCLASSINFOpdf
\else
\fi

\usepackage{amsthm}
\newtheorem{lemma}{Lemma}
\newtheorem{remark}{Remark}

\usepackage{amsmath,amssymb,mathtools,dsfont,bm}

\usepackage{color}

\usepackage{algorithm,algorithmic}

\DeclarePairedDelimiter{\norm}{\lVert}{\rVert}

\usepackage{subcaption}

\usepackage{amsmath,graphicx,float,cite,dsfont,amssymb}
\usepackage{psfrag}
\usepackage{color}
\usepackage{pgfplots}
\usepackage[utf8]{inputenc}
\usepackage[english]{babel}

\usepackage{tikz}
\usetikzlibrary{arrows}
\usetikzlibrary{patterns}
\tikzset{>=latex'}
\usetikzlibrary{shapes.geometric}
\usetikzlibrary{calc}
\usetikzlibrary{positioning, shapes,fadings,through}

\definecolor{color1}{rgb}{0.0503832136000000,	0.0298028976000000,	0.527974883000000}%
\definecolor{color2}{rgb}{0.493426212287426,	0.0115244183443711,	0.658032128750000}%
\definecolor{color3}{rgb}{0.796382339000000,	0.277979917955297,	0.471322439808427}%
\definecolor{color4}{rgb}{0.972230291500000,	0.581718982058214,	0.254084505027532}%
\definecolor{color5}{rgb}{0.940015097000001,	0.975158357000000,	0.131325517000000}%

\usetikzlibrary{decorations}
\pgfdeclaredecoration{lightning bolt}{draw}{
	\state{draw}[width=\pgfdecoratedpathlength]{
		\pgfpathmoveto{\pgfpointorigin}%
		\pgfpathlineto{\pgfpoint{\pgfdecoratedpathlength*0.6}%
			{-\pgfdecoratedpathlength*.1}}%
		\pgfpathlineto{\pgfpoint{\pgfdecoratedpathlength*0.55}{0pt}}%
		\pgfpathlineto{\pgfpoint{\pgfdecoratedpathlength}{0pt}}%
		\pgfpathlineto{\pgfpoint{\pgfdecoratedpathlength*0.4}%
			{\pgfdecoratedpathlength*.1}}%
		\pgfpathlineto{\pgfpoint{\pgfdecoratedpathlength*0.45}{0pt}}%
		\pgfpathclose%
	}%
}

\usepackage{xcolor}
\begin{document}
	%
	\title{Immersive Virtual Realities via Cooperative NOMA in Hybrid Cloud/Mobile Edge Computing Networks}
	\title{Extended Reality via Cooperative NOMA in Hybrid Cloud/Mobile-Edge Computing Networks}
	
	\author{\IEEEauthorblockN{Robert-Jeron Reifert,~\IEEEmembership{Graduate Student Member,~IEEE,}
			Hayssam Dahrouj,~\IEEEmembership{Senior Member,~IEEE,}\newline and
			Aydin Sezgin,~\IEEEmembership{Senior Member,~IEEE}}
		\thanks{%
			\ifarxiv%
			This paper is accepted for publication at the IEEE Internet of Things Journal.\newline
			\else%
			\fi%
			This work was supported in part by the German Federal Ministry of Education and Research (BMBF) in the course of the 6GEM Research Hub under grant 16KISK037.\newline
			\IEEEauthorblockA{Robert-Jeron Reifert and Aydin Sezgin are with Digital Communication Systems, Ruhr University Bochum, Bochum, Germany. (Email: \{robert-.reifert,aydin.sezgin\}@rub.de).\newline
				Hayssam Dahrouj is with the Department of Electrical Engineering, University of Sharjah, Sharjah, United Arab Emirates. (Email: hayssam.dahrouj@gmail.com).}\\
		}
}

\maketitle

\begin{abstract}
Extended reality (XR) applications often perform resource-intensive tasks, which are computed remotely, a process that prioritizes the latency criticality aspect. To this end, this paper shows that through	leveraging the power of the central cloud (CC), the close proximity of edge computers (ECs), and the flexibility of uncrewed aerial vehicles (UAVs), a UAV-aided hybrid cloud/mobile-edge computing architecture promises to handle the intricate requirements of future XR applications. 
In this context, this paper distinguishes between two types of XR devices, namely, strong and weak devices. The paper then introduces a cooperative non-orthogonal multiple access (Co-NOMA) scheme, pairing strong and weak devices, so as to aid the XR devices quality-of-user experience by intelligently selecting either the direct or the relay links toward the weak XR devices. 
A sum logarithmic-rate maximization problem is, thus, formulated so as to jointly determine the computation and communication resources, and link-selection strategy as a means to strike a trade-off between the system throughput and fairness. Subject to realistic network constraints, e.g., power consumption and delay, the optimization problem is then solved iteratively via discrete relaxations, successive-convex approximation, and fractional programming, an approach which can be implemented in a distributed fashion across the network. 
Simulation results validate the proposed algorithms performance in terms of log-rate maximization, delay-sensitivity, scalability, and runtime performance. The practical distributed Co-NOMA implementation is particularly shown to offer appreciable benefits over traditional multiple access and NOMA methods, highlighting its applicability in decentralized XR systems.
\end{abstract}
\begin{IEEEkeywords}
Extended reality, cooperative NOMA, central cloud, cloud computing, mobile-edge computing, hybrid networks, uncrewed aerial vehicles.
\end{IEEEkeywords}

%
\IEEEpeerreviewmaketitle

\section{Introduction}
The technical requirements for the sixth generation (6G) of wireless communication systems, particularly the vision of extended reality (XR) and their conjunction with emerging technologies, such as the metaverse, become a major driver in developing sophisticated techniques for the wireless physical-layer, access mechanisms, computing infrastructures, and resource management schemes \cite{8766143,8869705,9806418,10070393}. For instance, protocol testing in medicine, extended teaching methods in education, and visualization or modeling in engineering require enhanced virtual and telepresence experiences \cite{9363888,waisberg2023apple}.
The intricate requirements of XR pose a burden on the underlying 6G infrastructure in terms of big data processing, computationally intensive video rendering, connectivity issues, e.g., cell-borders coverage, and latency-criticality \cite{10033423}. For example, a latency beyond $10$~ms between control input and the corresponding video effect may cause motion sickness to XR users. To this end, this paper tackles computation, connectivity, and latency issues via 
a framework for managing XR devices quality-of-service by means of optimizing the network load-balancing in a hybrid cloud/mobile-edge computing (MEC) architecture. 

More precisely, given the trend of proliferated mobile devices, e.g., head-mounted displays for XR, each having a variable level of processing capabilities and battery life, 6G networks aspire to provide a multi-purpose infrastructure that is capable of accommodating the multiple system requirements. On one hand, such systems are expected to connect devices and provide service levels within the application's requirements. On the other hand, such systems promise to enable computational offloading to process complicated tasks remotely \cite{9830429}.
Along the same lines, cloud computing at the central cloud (CC) provides powerful computation resources, but introduces relatively large connection delays, which forms a \emph{single-point-of-failure}. Edge computers (ECs), on the other hand, are deployed in close proximity to the edge users, thereby reducing connection latency, albeit providing less computational resources with stricter power constraints \cite{8951269}. Combining the advantages of the CC and the ECs through a hybrid CC/MEC network architecture, therefore, promises to offer valuable joint computational and communication capabilities \cite{10034763,emr,8612452,7996346}.
For instance, providing connectivity to poorly served areas often calls for the use of uncrewed aerial vehicles (UAVs) so as to extend the radio access network (RAN) borders, which is bound to boost the mobile EC prospects \cite{9461747,9738797,pham2023emergence}.

From a multiple access scheme perspective, the conventional orthogonal multiple access (OMA) methods, e.g., based on the time or frequency domain, often fail to support the large number of devices in dense networks \cite{7973146}. For instance, the use of beamforming via multiple antennas, referred to as space division multiple access (SDMA), is capable of serving devices simultaneously; however, the number of served links in SDMA is limited by the number of deployed antennas. Non-orthogonal multiple access (NOMA) becomes, thus, an alternative to serve multiple devices through sharing the same time-frequency resource, supporting device numbers beyond the limitation of OMA \cite{7973146}.
Recently, cooperative non-orthogonal multiple access (Co-NOMA) emerged as a cooperative scheme for enhancing the system capacity, spectral efficiency, reduced latency, reliability, flexibility, and user fairness \cite{9143270}. Co-NOMA extends NOMA by means of including the possibility of device cooperation. Since blockages cause adverse channel conditions for some devices, more equitable performance is enabled by device relaying in Co-NOMA, which motivates its adoption in XR systems.

Along such lines, this paper considers a hybrid CC/MEC system, where the CC connects to base stations (BSs) through capacity-limited fronthaul links, and where the UAVs aim at boosting the edge connectivity. The network then aims at serving the XR devices, which in turn are categorized as either strong or weak devices. The paper then employs a Co-NOMA framework within such a UAV-aided hybrid network to specifically enable the 6G use-cases of XR, which relies on joint communication and computations resource allocation. In XR, control inputs from the XR device, e.g., field of view and movements, are sent to the CC or an EC\footnote{UAVs and ECs are interchangeably used throughout the manuscript.}. Then the XR scene is rendered and encoded into a video stream, which is transmitted back to the XR device. In an effort to provide a digitally equitable service to the weak and strong XR devices, this paper becomes the first of its kind that considers maximizing the network-wide sum-of-logarithmic-rates (log-rate) utility in a Co-NOMA-empowered XR setup using a hybrid CC/MEC network. The paper, particularly, focuses on optimizing the computation allocation, physical layer resources (i.e., beamforming vectors and power allocations), and link layer resources (i.e., link-selection, clustering, and time-split factors). The problem is formulated subject to realistic network constraints including computation capacity, physical layer constraints (i.e., fronthaul capacity, achievable rate, and power consumption), and overall per-device delay constraints.

\subsection{Related Work}
We next overview topics related to this paper in terms of cloud-, MEC-, and UAV-aided networks, interference management adopting SDMA, NOMA, and Co-NOMA, as well as XR systems design and resource management.
\textbf{CC/MEC networks:} Jointly managing communication and computation resources in MEC networks is a timely subject, that is considered in \cite{8016573,10033423}. While \cite{8016573} surveys the communication perspective, \cite{10033423} considers an explicit XR use-case. The proposed reinforcement learning-based scheme in \cite{10033423} achieves millisecond-latency and deals with the heavy computation burden of 360-degree video streaming. 
In emergency response \cite{9597869}, BS breakdown scenarios \cite{8757041}, harsh	factory environments \cite{8647441}, or cell-edges \cite{8959360}, UAVs particularly come into play so as to ensure opportunistic connectivity \cite{8434285,8764580}. Specifically, in production lines with multiple machines, large-scale depots, and remote-assistance scenarios, UAVs can aid the XR connectivity. For instance, a computation rate maximization problem is considered for a UAV-aided mobile-edge computing network in \cite{8434285}. Contrary to the single UAV setup, multiple UAVs are considered in \cite{8764580} for energy efficient resource allocation by optimizing user-association, computation rate allocation, and location planning.
None of the above works, however, consider the specific interplay of UAVs and the CC. The works \cite{8885877,8664595} show that a network with CC and EC coexistence is able to effectively minimize the overall latency, which is particularly important for XR applications in order to provide an immersive experience.
All of the previous schemes (except the one described in \cite{10033423}) rely on time-domain OMA, and are not suitable to effectively serve dense networks. Further, the above schemes do not consider the multi-user, inter-UAV, and CC-EC interference management aspect.

\textbf{Multiple access techniques:} For interference management, SDMA relies on multi-antenna beamforming to separate users in the spatial domain. For example, the work \cite{10008481} considers a joint communication and computation resource allocation using SDMA in hybrid cloud/mobile-edge computing networks, dealing with multi-user, inter-UAV, and CC-EC interference. Reference \cite{10008481}, in particular, focuses on maximizing a weighted-sum rate, and shows how the performance is strongly coupled with the number of available antennas. NOMA, on the other hand, employs successive interference cancellation, and is shown to achieve enhanced performance metrics, e.g., maximizing the sum computation rate \cite{9461747}, minimizing the energy consumption \cite{8543183}, and enabling cooperative computation \cite{8951269}. References \cite{9461747,8543183,8951269}, however, do not capture the impact of NOMA on hybrid CC/MEC networks.
To this end, along with reference \cite{10033423}, the work \cite{10034763} considers the rate-splitting (RS) paradigm as a generalization of NOMA. While RS is shown to both enable XR services, and outperform the state-of-the-art schemes under strict power and delay constraints, both works \cite{10033423} and \cite{10034763} do not consider the potential benefits of device cooperation. 
Such aspect is, in fact, of measurable importance especially in situations where some of the devices are blocked. To this end, Co-NOMA goes beyond NOMA and RS, emerging as a solution that exploits the device-to-device (D2D) link in order to enhance the service of weakly connected devices \cite{s23083958,9259258,7445146}. While references \cite{s23083958,9259258,7445146} emphasize the Co-NOMA communication performance improvement, the analysis of the computation and delay aspects therein remains an open problem.
More recently, in a single-EC setup, some works consider Co-NOMA-enabled MEC for minimizing energy consumption \cite{9325063}, for maximizing computation efficiency \cite{9345931}, and in a cognitive radio setup \cite{9749756}.
An end-to-end delay analysis of a relaying network is considered in \cite{7342977}. The impact of Co-NOMA in hybrid networks consisting of CC and multiple UAV-enabled ECs remains, however, majorly unstudied.

\textbf{XR system design:} From an XR system design perspective, the recent literature touches upon both the prospects and challenges of XR systems, e.g., an overview of the XR technology can be found in \cite{kovacova2022immersive}, the interplay with edge computing in 6G is discussed in \cite{10.1007/978-981-15-3075-39}, and various enabling architectures are proposed in \cite{10033423,9363888,9894514,8612452,chaccour2023joint}. \cite{9363888}, particularly, describes computationally exhaustive tasks, the battery limitations, and the stringed communication requirements, such as reliability and latency, as major technical challenges of XR. In an effort to ease both computation and power requirements at the XR device, cloud computing promises to broaden the applicability of XR.
A generalized view of an XR system within a cloud-based framework is given in \cite{9363888}, which considers an adaptive video bitrate to cope with mixed-critical network traffic. 
The work \cite{9894514} covers the edge computing aspect in an XR system by tackling the CC connection-delay. Reference \cite{9894514}, particularly, optimizes the execution delay and energy consumption using a deep reinforcement learning approach for managing the offloading decisions. In an effort to combine the cloud and edge computing perspectives, \cite{8612452} considers the interplay of user-layer, edge-layer, and cloud-layer in a hierarchical computation architecture, revealing improved energy and delay performance.
Yet, advanced wireless communication methods, e.g., beamforming or Co-NOMA, and the relevant impact of hybrid networks remain open issues for future XR networks. To this end, this paper focuses on offloading computationally intensive tasks to the hybrid CC/MEC network, and exploring how Co-NOMA can tackle parts of the reliability and latency requirements of XR systems.

\subsection{Contribution}
Unlike the aforementioned references, this paper considers a hybrid network consisting of multi-antenna BSs and UAVs, which collaboratively serve a number of XR devices. The CC manages the core-network functions and connects to the BSs via fronthaul links. The UAVs, on the other hand, perform communication, computation, and resource management at the network edge as MEC decentralized platforms. Further, to jointly enhance the system fairness and cope with the adverse channel conditions, we employ a Co-NOMA framework, where the strongly connected XR devices (hereafter denoted by strong devices) may aid the communication toward the weakly connected XR devices (hereafter denoted by weak devices) using D2D relaying. 

The paper then considers the problem of finding a reasonable trade-off between throughput and fairness, i.e., it maximizes the network-wide log-rate of the considered hybrid CC/MEC network. The paper then manages the joint communication and computation resources, e.g., data rate, beamforming vectors, link-selection, clustering, direct/relay time-split factors, and computation allocation subject to XR use-case considerations, namely, computation capacity, power, link-selection, achievable rate, and per-device delay constraints. The paper specifically solves such a non-convex, mixed-integer log-rate maximization problem using $\ell_0$-norm relaxation, successive-convex approximation (SCA), and fractional programming (FP), resulting in a computationally tractable iterative algorithm. 
The algorithm can conveniently be implemented in a decentralized fashion, by properly coordinating the resource management among the individual computation platforms (CC and UAVs). Numerical simulations validate the enhanced performance of Co-NOMA over SDMA and NOMA in terms of log-rate, delay, and scalability. The results of the paper particularly show how the direct/relay link selection and time-split factors have a strong impact on the system performance. The distributed resource management (DRM) is particularly shown to approach the centralized resource management (CRM) performance, all while depicting a decent convergence behavior, runtime, and practical applicability. 
The contributions of this paper can be summarized as follows:
\begin{itemize}
\item[1)] {\textit{UAV-aided hybrid cloud/mobile-edge computing network:}}
By means of combining the strong computation capabilities of the CC, the reduced connection delay of ECs, and the flexibility of UAVs, this paper proposes a UAV-aided hybrid CC/MEC framework to tackle the XR requirements in terms of latency-criticality and computation intensity perspectives. The considered network performance strongly relies on the appropriate joint management of communication and computation resources, i.e., beamforming vectors, transmit powers, allocated data rates, and computation rates. That is, the paper considers that the CC is connected to the BSs via capacity-limited fronthaul links, which adds to the overall delay metric. Further, the UAVs are assumed to be battery-limited, resulting in strict 
power limitations.
\item[2)] \textit{Co-NOMA mechanism for equitable XR devices:}
The paper system model allows a flexible utilization of the Co-NOMA framework. For instance, the strong devices, e.g., XR headsets, are served directly via BSs or UAVs, whereas the weak devices can be either served directly or through D2D relaying via their associated strong device. In addition to the link-selection, the communication time-split factor, i.e., the time allocation to both the direct service and the relaying service, is another system parameter that is carefully optimized in the paper.
The considered objective herein incorporates the fairness aspect by maximizing a log-rate function, which strikes a trade-off between the throughput and fairness measures.
\item[3)] \textit{Distributed optimization framework:}
Such log-rate maximization problem is formulated subject to computation capacity, link-selection, achievable rate, fronthaul capacity, maximum transmit power, power consumption, and per-device delay constraints. The paper then solves such a mixed-integer, non-convex optimization problem iteratively using discrete relaxations, SCA, and FP. One highlight of the proposed algorithm is that it can be implemented in a distributed fashion across the CC and the multiple EC platforms.
\item[4)] \textit{Numerical simulations:}
The impact and interplay of different network parameters on the proposed framework are analyzed via expressive numerical simulations. The CRM and DRM implementation of Co-NOMA are extensively compared to the centralized and distributed SDMA and NOMA schemes.
Specifically, the flexibility of the Co-NOMA scheme to adapt to different network scenarios and deployment cases is shown to enhance the total log-rate performance, i.e., the throughput, fairness, delay, and computational aspects. Further, the DRM is shown to perform close to the centralized scheme, to outperform the benchmarks in various setups, and to exhibit a valuable scalabilty in terms of runtime in dense networks.
\end{itemize}

\subsection{Organization}
In Section~\ref{sec:sysmod}, the XR setup, network setup, and computation and communication architecture are presented. The optimization problem is then formulated in Section~\ref{sec:prob_section}. The solution approach to solve such problem is described in Section~\ref{sec:alg}, which presents the reformulation steps, outlines the convex problem, and presents the distributed resource management algorithm. Simulation results are then presented in Section~\ref{sec:sim}. At last, we conclude the paper in Section~\ref{sec:con}.

\section{System Model}\label{sec:sysmod}
The system model considered in this paper is related to three parts, namely, the XR setup and computation architecture, the network setup including topology and association, and the transmission scheme based on Co-NOMA.
For accessibility, Tab.~\ref{tb:symbollist} provides a list of the symbols used in this manuscript. 
\begin{table}[t]
\renewcommand{\arraystretch}{1.0}
\centering
\begin{tabular}{c l}
	\hline
	Symbol & Description/meaning\\\hline
	$a$ & Quadratic transform auxiliary variable\\
	$b$ & BS index\\
	$c$ & CC index\\
	$d$ & Co-NOMA device pairing, data size\\
	$e$ & EC index\\
	$f$ & Computation allocation\\
	$g$ & Strong-to-weak device channel\\
	$h$ & BS-to-device and UAV-to-device channels\\
	$i$ & Mathematical auxiliary index\\
	$j$ & Mathematical auxiliary index\\
	$k$ & Device index\\
	$l$ & Antenna number/index\\
	$n$ & Noise\\
	$o$ & Computation transform parameter\\
	$p$ & Transmit powers\\
	$q$ & Aggregate beamforming vectors\\
	$r$ & Rates\\
	$s$ & Strong device index\\
	$t$ & Delay\\
	$u$ & CC beamformers\\
	$v$ & EC beamformers\\
	$w,W$ & Weak device index, transmit bandwidth\\
	$x$ & Device signal\\
	$y$ & Received signal\\
	$z$ & BS-device/link selection auxiliary variable\\\hline
	$\alpha$ & Lemma~\ref{lma3} quadratic transform auxiliary variable\\
	$\beta$ & Weights $\ell_1$-norm\\
	$\gamma,\Gamma$ & SINR variable, SINR\\
	$\delta$ & Normalization constant\\
	$\eta$ & Computational complexity\\
	$\theta,\Theta$ & Elevation, delay\\
	$\Lambda$ & Fronthaul delay\\
	$\mu$ & Computation transform parameter\\
	$\nu$ & Time-slot fraction\\
	$\sigma,\Sigma$ & Noise, sum\\
	$\tau$ & Lemma~\ref{lma1} and Lemma~\ref{lma3} auxiliary variable\\
	$\Upsilon$ & Direct or relay link index\\
	$\Phi$ & Operational power\\
	$\chi$ & Lemma~\ref{lma1} auxiliary function\\
	$\Psi$ & Lemma~\ref{lma1} and Lemma~\ref{lma3} auxiliary function\\
	$\omega,\Omega$ & Convergence threshold, worst-case iterations\\\hline
\end{tabular}
\vspace*{-.15cm}
\caption{List of symbols.}
\label{tb:symbollist}
\end{table}%

\subsection{XR Setup and Computation Architecture}
\begin{figure}[!t]
\centering
\includegraphics[width=\linewidth]{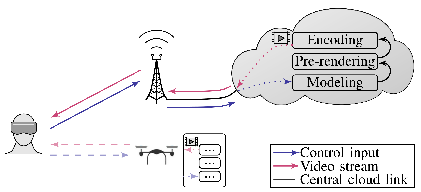}
\vspace*{-.2cm}
\caption{Example XR system consisting of a head-mounted display, a BS (alternatively a UAV), and the CC (alternatively an on-chip EC).}
\label{xr_mdl}
\end{figure}
In the considered XR setup, we assume that XR devices, such as head-mounted displays, tablets, smartphones, smartwatches, etc., require specific computation services offered by the CC or EC infrastructure. An illustrative example of such system is shown in Fig.~\ref{xr_mdl}, where a head-mounted display receives a 360-degree video stream, which is pre-rendered and encoded at the CC. In parallel, control input is taken from the XR device. That is, movement data or field-of-view impact the XR model. Hence, modeling, pre-rendering, and encoding become highly coupled at the cloud \cite{9363888}. This work also assumes that the control input requires transmission of data in orders of magnitude less than the 360-degree video stream, which allows to consider the computation and the downlink transmission aspects jointly \cite{10033423}.
Since we assume that the BSs cover the core-network part, i.e., every XR device within the RAN borders, some devices might be positioned at or beyond the cell-edge. As an alternative to being served by the CC, the XR device can be served by a UAV with on-chip EC, which offers better connectivity links, albeit using lower computation capabilities than the CC.

Each network device requires tasks, e.g., field-of-view pre-rendering, tracking, or detection \cite{10033423}, to be computed remotely at the corresponding computing platform, i.e., either at the CC or at one of the ECs. The task of device $k$ requires a specific amount of computation cycles $F_k$. Upon completing the computation task, each platform sends the data output to the associated devices. Let the \emph{computation vector} $\bm{f}\in\mathbb{N}^K$ denote the allocated computation cycles to compute the devices' tasks. In other words, each computation platform assigns $f_k$ cycles/s for computing device $k$'s task, where $k$ is an associated device.
Each platform is subject to the following maximum computation capacity constraint
\begin{equation}
\sideset{}{_{k\in\mathcal{K}_\text{c}}}\sum f_{k} \leq f_{c}^{\text{max}}; \quad \sideset{}{_{k\in\mathcal{K}_e}}\sum f_{k} \leq f_{e}^{\text{max}}, \; \forall e\in\mathcal{E}, \label{eq:maxcompcap}
\end{equation}
where $f_{c}^{\text{max}}$ and $f_{e}^{\text{max}}$ are the nominal maximum computation capacities of the CC and EC $e$, respectively.

\subsection{Network Setup}
The hybrid CC/MEC network considered in this paper consists of one central cloud and $E$ mobile-edge computing components. The CC includes a central processor (CP) responsible for managing the CC's operation in terms of computation and communication. That is, the CP performs computation of the core-devices' tasks and assigns the communication resources, i.e., the CP performs most network functions, e.g., encoding and precoder design \cite{7809154}. The network includes $B$ multi-antenna BSs, where BS $b$ is connected to the CC via a fronthaul link of limited capacity $R_b^\text{max}$. Such BSs serve the XR core-devices positioned within the RAN borders. At the network edge, the EC components are then implemented on multi-antenna UAVs, with on-chip computation capabilities, so as to serve the edge-devices. Each BS $b$ has $L_b$ antennas, and each UAV $e$ has $L_e$ antennas.
The set of ECs is denoted by $\mathcal{E}=\{1,\cdots,E\}$, the set of BSs by $\mathcal{B}=\{1,\cdots,B\}$, and the set of $K$ single-antenna XR devices by $\mathcal{K}=\{1,\cdots,K\}$. In the context of CC/EC coexistence, we denote core-devices located in the central network part and associated with the CC as $\mathcal{K}_\text{c}\subseteq\mathcal{K}$, and the edge-devices associated with EC $e$ as $\mathcal{K}_e\subseteq\mathcal{K}$. Such sets are disjoint, where a device $k$ can either be associated with the CC or one EC, i.e., $\mathcal{K}_\text{c} \cap \mathcal{K}_e  = \emptyset, \forall e\in\mathcal{E}$ and $\mathcal{{K}}_e \cap \mathcal{{K}}_e'  = \emptyset, \forall e\neq e'$. In this work, the sets $\mathcal{K}_e$ and $\mathcal{K}_c$ are chosen based on the device locations. Fig.~\ref{sys_mdl} illustrates such a network which consists of $2$ BSs serving $3$ core-device pairs, and $2$ UAVs, each serving one device pair. 
\begin{figure}[!t]
\centering
\includegraphics[width=\linewidth]{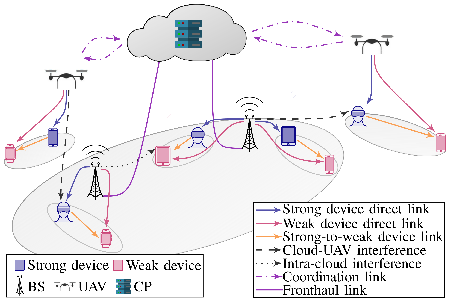}
\vspace*{-.2cm}
\caption{Hybrid network with $2$ UAVs, $2$ BSs, and $10$ devices.}
\label{sys_mdl}
\end{figure}


\subsection{Communication Architecture}
After the first phase of computing the tasks, the resulting data (e.g., a 360-degree video stream \cite{10033423}) is sent to the devices under a downlink Co-NOMA-based transmission scheme. Within the considered XR setup, there are multiple connected devices, each experiencing some channel conditions toward the BSs or UAVs. For example, the XR headsets act as strong devices due to the higher line-of-sight probabilities toward the BSs and UAVs. Smartwatches, on the other hand, tend to be blocked by the human body, or experience high-mobility in an unpredictable manner, and can thus be treated as weak devices. In fact, such classification highlights the need to tackle blockage or adverse channel conditions on BS/UAV links using direct D2D transmissions. Depending on the channel strengths or proximity to the radio transmitters, the devices are paired as strong and weak devices, i.e., similar to \cite{9259258}. The paper then utilizes a time-division multiplexing scheme where, in the first time slot, the network serves the strong devices and a subset of weak devices. In the second time slot, the BSs and the UAVs remain idle, while some strong devices forward the weak devices' messages decoded in the first time slot to their associated weak devices.

Due to the dynamic nature of communication networks, the channel condition of weak devices varies over time. 
This paper, therefore, proposes a hybrid transmission scheme, where the link selection, i.e., choosing whether the weak device would be served by the strong devices or directly by the CC or one of the ECs, is part of the optimization problem.
Let the set of strong devices be $\mathcal{K}^\text{str}\subseteq\mathcal{K}$, and let the set of weak devices be $\mathcal{K}^\text{weak}\subseteq\mathcal{K}$. In this context, the core-devices can figuratively lay within the proximity of the BSs. Edge-devices, on the other hand, lay beyond the core-network borders. For the ease of presentation, devices are assumed to be generated in pairs, where the channel quality determines whether a device is denoted as strong or weak \cite{e19070362}, see Fig.~\ref{sys_mdl}. Note that this paper optimizes the resource allocation within resource-blocks where the channel conditions remain constant, a routine that can be called whenever the channel changes. 

\subsection{Received Signal Model}
At the CC, the data initially needs to be forwarded over the fronthaul links to the BSs, which in turn perform modulation, precoding, and specific radio tasks. At the edge side, the ECs forgo the need of an additional fronthaul link and directly transmit the data to their respective devices. As such, the message to each device is encoded into signals $x_k$, where $x_k\sim\mathcal{CN}(0,1)$, $\forall k\in\mathcal{K}$, and all the $x_k$'s are independent and identically distributed (i.i.d.), circularly symmetric random variables. 

Due to the diverse nature of communication links across the network, this paper defines the following list of channel vectors: $(i)$ The channel vector from BS $b$ to device $k$ is $\bm{h}_{b,k}\in\mathbb{C}^{L_\text{b}}$; $(ii)$ The channel vector from UAV $e$ to device $k$ is $\bm{\tilde{h}}_{e,k}\in\mathbb{C}^{L_\text{e}}$; $(iii)$ The channel value from strong device $s$ to weak device $w$ is ${g}_{s,w}\in\mathbb{C}$. In the first time slot, we denote the aggregate channel vector of device $k$ by $\bm{h}_{k} = [\bm{h}_{1,k}^T,\cdots,\bm{h}_{B,k}^T,\tilde{\bm{h}}_{1,k},\cdots,\tilde{\bm{h}}_{E,k}]^T$. We assume full knowledge of the channel state information at the transmitters (CSIT) for both the BS-device and UAV-device links. 

We also let $\bm{u}_{b,k}^{(d)}\in\mathbb{C}^{L_\text{c}}$ and ${\bm{v}}_{e,k}^{(d)}\in\mathbb{C}^{L_\text{e}}$ be the \emph{beamforming vectors} of device $k$'s signal at BS $b$, and at UAV $e$, respectively, where the superscript $(d)$ denotes the \emph{direct} link between each transmitter and its respective device. Similarly, the vectors $\bm{u}_{b,k}^{(r)}\in\mathbb{C}^{L_\text{c}}$ and ${\bm{v}}_{e,k}^{(r)}\in\mathbb{C}^{L_\text{e}}$ are used to refer to the Co-NOMA-enabled \emph{relayed} beamforming vectors, i.e., the beamforming vectors from each transmitter (CC or EC) to the strong device devoted to transmit the message intended for the weak device. In this paper, only weak devices are served by either a direct or relayed link, thus all strong devices $s\in\mathcal{K}^\text{str}$ are served directly by a CC or one of the ECs, i.e., $\bm{u}_{b,s}^{(r)}=\bm{0}_{L_\text{c}}$ and ${\bm{v}}_{e,s}^{(r)}=\bm{0}_{L_\text{e}}$.
For any device $k$, let the \emph{direct} aggregate beamforming vector be $\bm{q}_{k}^{(d)} = [(\bm{u}_{1,k}^{(d)})^T,\cdots,(\bm{u}_{B,k}^{(d)})^T,({\bm{v}}_{1,k}^{(d)})^T,\cdots,({\bm{v}}_{E,k}^{(d)})^T]^T$ and the \emph{relay}-intended aggregate beamforming vector be $\bm{q}_{k}^{(r)} = [(\bm{u}_{1,k}^{(r)})^T,\cdots,(\bm{u}_{B,k}^{(r)})^T,({\bm{v}}_{1,k}^{(r)})^T,\cdots,({\bm{v}}_{E,k}^{(r)})^T]^T$. Thus, the set of all beamforming vectors is given by $\bm{q}=\big[(\bm{q}_{1}^{(d)})^T,\cdots,(\bm{q}_{K}^{(d)})^T,(\bm{q}_{1}^{(r)})^T,\cdots,(\bm{q}_{1}^{(r)})^T\big]^T$.

Additionally, since each weak device is served either via the direct link or via the relay link, we get the following constraint:
\begin{equation}
\norm[\big]{\norm[\big]{\bm{{q}}_{w}^{(d)}}_{2}^{2}}_0 + \norm[\big]{\norm[\big]{\bm{{q}}_{w}^{(r)}}_{2}^{2}}_0 \leq 1, \quad \forall w\in\mathcal{K}^\text{weak},\label{eq:dir_rel_decision}
\end{equation}
\begin{figure*}[t]
\normalsize
\begin{align}
	y_s^{(1)} &= \bm{h}_s^H {\bm{q}_s^{(d)}} x_s + \hspace*{-.2cm}&\bm{h}_s^H (\bm{q}_{d_s}^{(d)}+\bm{q}_{d_s}^{(r)}) x_{d_s} &+ \sideset{}{_{s'\in\mathcal{K}^{\text{str}}\backslash\{s\}}}\sum \hspace*{-.25cm}&\bm{h}_{s}^H \bm{q}_{s'}^{(d)} x_{s'} &+ \sideset{}{_{w\in\mathcal{K}^\text{weak}\backslash\{d_s\}}}\sum \bm{h}_s^H (\bm{q}_{w}^{(d)}+\bm{q}_{w}^{(r)}) x_{w} \hspace*{.1cm}+ n_s \label{eq:y_m}\\
	y_w^{(1)} &= &\bm{h}_w^H (\bm{q}_w^{(d)}+\bm{q}_w^{(r)}) x_w &+ \sideset{}{_{s\in\mathcal{K}^{\text{str}}}}\sum &\bm{h}_{w}^H {\bm{q}_{s}^{(d)}} x_{s} &+ \sideset{}{_{w'\in\mathcal{K}^\text{weak}\backslash\{w\}}}\sum \bm{h}_w^H (\bm{q}_{w'}^{(d)}+\bm{q}_{w'}^{(r)}) x_{w'} + n_w\label{eq:y_n1}
\end{align}
\hrulefill
\vspace*{-.4cm}
\end{figure*}%
\noindent where the $\ell_0$-norm in \eqref{eq:dir_rel_decision} is used to determine whether any power is assigned to serve device $w$ using the direct or relayed link. 
At the D2D transmission during the second time slot, only strong devices $s\in\mathcal{K}^\text{str}$ have a non-zero transmission power. Let $\bm{p} = [p_1,\cdots,p_K]$ be the \emph{transmit power vector} across all users. The received signal at the strong device $s$ and the weak device $w$ in the first time slot can then be respectively formulated as \eqref{eq:y_m} and \eqref{eq:y_n1} (shown on the top of the next page), where $d_s$ denotes the weak device associated to $s$. The Co-NOMA device pairing is thus captured in $\bm{d}=[d_1,\cdots,d_{K}]$, where $d_k$ is the associated weak (strong) device if $k\in\mathcal{K}^\text{str}$ (if $k\in\mathcal{K}^\text{weak}$).

In the second time slot, the received signal at the weak device $w$ can be written as:
\begin{align}
y_{w}^{(2)} &= g_{d_w,w} \sqrt{p_{d_w}} x_w + \sum_{w'\in\mathcal{K}^\text{weak}\backslash\{w\}} g_{d_{w'},w} \sqrt{p_{d_{w'}}} x_{w'} + n_{w},\label{eq:y_n}
\end{align}
where $d_w$ denotes the strong device associated to $w$, 
and ${p_{d_w}}$ is the transmit power of device $d_w$. 

\section{Problem Formulation}\label{sec:prob_section}
This paper considers the problem of maximizing the sum log-rate across the considered hybrid CC/MEC network over the two Co-NOMA time slots presented above. In the sequel, we first present the achievable rates across the two slots, the associated joint computation and communication constraints, and the mathematical formulation of the considered problem.

\subsection{Achievable Rates}
In the first time slot, the strong devices receive their desired signals, as shown in \eqref{eq:y_m}. Based on \eqref{eq:y_m} (shown on top of the next page), we observe that each device receives its own intended signal $x_s$, its associated weak device's signal $x_{d_w}$, as well as other interfering signals and noise.
In order to define a unified expression for the interference terms, let $i$ denote either a strong or a weak device, where $i\in\mathcal{K}$. Further, let the set $\mathcal{X}_i$ be the set of strong devices which interfere with device $i$. Also, let $\mathcal{Y}_i$ be the set of interfering weak devices at device $i$. Then the interference plus noise experienced by device $i$, defined as $I_i(\mathcal{X}_i,\mathcal{Y}_i)$, can be written as:
\begin{align}
I_i(\mathcal{X}_i,\mathcal{Y}_i) &= \sigma^2 + \sum_{i'\in\mathcal{X}_i}\big|\bm{h}_i^H\bm{q}_{i'}^{(d)}\big|^2 \nonumber\\
&\hspace*{1.8cm}+ \sum_{j'\in\mathcal{Y}_i}\big|\bm{h}_i^H\big(\bm{q}_{j'}^{(d)}+\bm{q}_{j'}^{(r)}\big)\big|^2, \label{eq:I_i}
\end{align}
where $\sigma^2$ is the additive white Gaussian noise variance. Note that in \eqref{eq:I_i}, the strong and weak device interfering sets $\mathcal{X}_i$ and $\mathcal{Y}_i$ depend on the decoding order at device $i$.
Following the Co-NOMA \emph{relay} strategy, first the message intended to weak device $w$ is decoded at its associated strong device $d_w$ by treating the interference as noise. The signal to interference plus noise ratio (SINR) of weak device $w$'s message at strong device $d_w$ at the first time slot becomes
\begin{align}
\Gamma_{w}^{(r)} &= |\bm{h}_{d_w}^{H} {\bm{q}_{w}^{(r)}}|^2/I_{•d_w}(\mathcal{K}^\text{str},\mathcal{K}^\text{weak}\backslash\{w\}).\label{eq:SINR_n1}
\end{align}%
Afterwards, device $s$ is intended to decode its own signal. However, when the associated weak device is served directly, $s$ does not decode its weak device's message, leaving $d_w$ as part of the interference. Such considerations are captured in the following SINR expression for $s$ decoding its own message
\begin{align}
\Gamma_{s} &= \begin{cases}
	\frac{|\bm{h}_{s}^{H} \bm{q}_{s}^{(d)}|^2}{I_s(\mathcal{K}^\text{str}\backslash\{s\},\mathcal{K}^\text{weak}\backslash\{d_w\})}, &\text{if } x_{d_w} \text{ is decoded},\\[.1cm]
	\frac{|\bm{h}_{s}^{H} \bm{q}_{s}^{(d)}|^2}{I_s(\mathcal{K}^\text{str}\backslash\{s\},\mathcal{K}^\text{weak})}, &\text{otherwise}.
\end{cases}\label{eq:SINR_m}
\end{align}
Similar to the above discussion, the weak device $w$, at the first time slot, receives its own intended signal $x_w$, the interfering signals, and noise, as shown in \eqref{eq:y_n1} (shown on top of the next page). At weak device $w$, the SINR of the direct transmission in time slot 1 can then be written as:
\begin{align}
\Gamma_{w}^{(d)} &= |\bm{h}_{w}^{H} \bm{q}_{w}^{(d)}|^2/I_w(\mathcal{K}^\text{str},\mathcal{K}^\text{weak}\backslash\{w\}).\label{eq:SINR_nd}
\end{align}
Now, in the second time slot, the strong devices forward the weak device's signals. Thus, the SINR of weak device $w$, served by strong device $d_w$ becomes
\begin{align}
\Gamma_{w}^{(2)} &= \frac{|g_{d_w,w}|^2 {p_{d_w}}}{\sigma^2 + {\sideset{}{_{w'\in\mathcal{K}^\text{weak}\backslash\{w\}}}\sum |g_{d_{w'},w}|^2 {p_{d_{w'}}} }}.\label{eq:SINR_n2}
\end{align}

Define the achievable rates of weak device $w$'s message for the following two cases, for all $w\in\mathcal{K}^\text{weak}$: $(a)$ If the Co-NOMA \emph{relaying} strategy is selected, the rate of $w$'s message decoded at the corresponding strong device in time slot $1$ is denoted as $r_{w}^{(r)}$ (hop one) and the rate of $w$'s message in time slot $2$ is denoted as $r_{w}^{(2)}$ (hop two); $(b)$ If the direct transmission is selected, $w$'s rate is defined as $r_{w}^{(d)}$. The achievable rate of device $s$ is defined as $r_s$, for all $s\in\mathcal{K}^\text{str}$. The rate expressions of $r_{w}^{(r)}$, $r_s$, $r_{w}^{(d)}$, and $r_{w}^{(2)}$ can then be written as follows:
\begin{align}
r_{w}^{(r)} &\leq \nu W\log_2(1+\Gamma_{w}^{(r)}), &&\forall w \in \mathcal{K}^\text{weak},\label{eq:rmn1}\\
r_s &\leq \nu W\log_2(1+\Gamma_{s}), &&\forall s \in \mathcal{K}^\text{str},\label{eq:rm}\\
r_{w}^{(d)} &\leq \nu W\log_2(1+\Gamma_{w}^{(d)}), &&\forall w \in \mathcal{K}^\text{weak},\label{eq:rnn1}\\
r_{w}^{(2)} &\leq (1-\nu) W\log_2(1+\Gamma_{w}^{(2)}), &&\forall w \in \mathcal{K}^\text{weak},\label{eq:rmn2}
\end{align}%
where $W$ is the transmission bandwidth, and $\nu$ is the time-split factor. As an example, if $\nu=\frac{1}{2}$, the CC/EC network is active half of the time, while the other half is reserved for the strong-to-weak device links. 

Let $\bm{r}\in\mathbb{R}^K$ be the \emph{rate allocation vector} which characterizes the rates of all devices, with $\bm{r} = [r_1,\cdots,r_K]$. Following \cite{9259258}, a \emph{selection combining}-based mechanism is applied at the weak device. That is, the weak device $w$'s rate is determined by the maximum of the direct link and the strong device relayed link. The link with higher rate is selected for device $w$'s data reception, which is formulated as follows
\begin{equation}
{r_w \leq \text{max}\left(r_{w}^{(d)},\text{min}\left(r_{w}^{(r)},r_{w}^{(2)}\right)\right).}
\end{equation}
The above can also be represented using the \emph{direct}-link and \emph{relay}-link beamforming vectors explicitly as follows:
\begin{equation}
{r_w \leq \norm[\big]{\norm[\big]{\bm{{q}}_{w}^{(d)}}_{2}^{2}}_0 r_{w}^{(d)} + \norm[\big]{\norm[\big]{\bm{{q}}_{w}^{(r)}}_{2}^{2}}_0 \text{min}\left(r_{w}^{(r)},r_{w}^{(2)}\right).}\label{eq:rnsc}
\end{equation}
Note that in \eqref{eq:rnsc}, the link is selected using the $\ell_0$-norm.
With these strong and weak device rates at hand, we note that under the considered hybrid CC/MEC network, the CC part, i.e., the network part containing the core-devices, comes with additional constraints. That is, the cloud-to-BS (fronthaul) connection is limited by the finite capacity fronthaul constraint
\begin{equation}
\sideset{}{_{k\in\mathcal{K}}}\sum \big( \norm[\big]{\norm[\big]{\bm{{u}}_{b,k}^{(d)}}_{2}^{2}}_0+\norm[\big]{\norm[\big]{\bm{{u}}_{b,k}^{(r)}}_{2}^{2}}_0 \big) r_{k} \leq R_b^\text{max}, \quad\forall b \in \mathcal{B}, \label{eq:front}
\end{equation}
where the $\ell_0$-norm determines whether BS $b$ assigns power to serve device $k$ or not.

\subsection{Joint Computation and Communication}
One of this paper goals is to jointly determine the computation and communication resources in the considered hybrid cloud/MEC network. In what follows, we elaborate on the coupled computation and communication constraints, especially those related to the power and delay constraints.
\subsubsection{Power Consumption}
We consider three power consumption metrics, $p_k$, $P_{b}^\mathrm{cc}$, and $P_{e}^\mathrm{ec}$, which denote strong device $k$'s, BS $b$'s, and EC $e$'s power consumptions, respectively. The latter two are given as:
\begin{align}
&P_{b}^\mathrm{cc} = \sideset{}{_{k\in\mathcal{K}_c}}\sum \norm[\big]{\bm{{u}}_{b,k}}_{2}^{2}, \label{eq:pcc} \\
&P_{e}^\mathrm{ec} = \underbrace{\sideset{}{_{k\in\mathcal{K}_e}}\sum\norm[\big]{{\bm{v}}_{e,k}}_{2}^{2}}_{\text{Transmission}} + \underbrace{o_e \left(\sideset{}{_{k\in\mathcal{K}_e}}\sum f_{k} \right)^{\mu_e}}_{\text{Computation}} +\underbrace{\phantom{\Big(}\Phi_e\phantom{\Big)}}_{\text{Operation}}, \label{eq:pec}
\end{align}
where $o_e$ and $\mu_e$ are CPU model constants \cite{8434285}. On the one hand, the CC is assumed to be restricted in transmission consumption, as shown in \eqref{eq:pcc}. The ECs, on the other hand, are subject to three types of power consumptions. More specifically, $P_{e}^\mathrm{ec}$ in \eqref{eq:pec} comprises transmission, computation, and operational power \cite{8434285}, where the operational power $\Phi_e$ is fixed and accounts for mechanical, flight, and operational power. Finally, the power constraint at the strong device is given by:
\begin{equation}
p_s \leq P_s^\text{max}, \quad\forall s\in\mathcal{K}^\text{str},\label{eq:psmax}
\end{equation}
where $P_s^\text{max}$ is the maximum transmit power of device $s$.
\subsubsection{Delay Considerations}
Device $k$'s delay is given by
\begin{align}
&\Theta_{k} = \underbrace{{F_k}/{f_{k}}}_{\text{Computation Delay}} + \underbrace{\phantom{/}\Lambda_{k}\phantom{/}}_{\text{Fronthaul Delay}} + \underbrace{{D_k}/{r_k}}_{\text{Transmission Delay}},\label{eq:delaycc}
\end{align}
where $F_k$, $\Lambda_{k}$, and $D_k$ are the computation cycles required for $k$'s task, the fronthaul delay, and the resulting data size, respectively. In other words, \eqref{eq:delaycc} comprises of computation delay, i.e., the task processing time, fronthaul delay, i.e., the time-loss on the cloud-BS link, and transmission delay, i.e., the physical wireless data transfer time-loss. The middle term in \eqref{eq:delaycc}, i.e., the fronthaul delay, is only relevant for core-devices $k\in\mathcal{K}_c$, and is rather set to zero for edge-devices, i.e., $\Lambda_{k}=0$, $\forall k\in\mathcal{K}_e, \forall e\in\mathcal{E}$.
For all core-devices, the fronthaul delay can be explicitly written as
\begin{align}
\Lambda_k = \underset{b\in\mathcal{B}_k}{\text{max}}\Bigg\{ \frac{1}{R_b^\text{max}} \sum_{i\in\mathcal{K}_c}&  \bigg( \norm[\big]{\norm[\big]{\bm{{u}}_{b,i}^{(d)}}_{2}^{2}}_0+\norm[\big]{\norm[\big]{\bm{{u}}_{b,i}^{(r)}}_{2}^{2}}_0 \bigg)  D_{i} \Bigg\},\label{eq:Lambdak}
\end{align}
where $\mathcal{B}_k$ is the set of BSs serving device $k$. \eqref{eq:Lambdak} utilizes the $\ell_0$-norm to determine which devices' data are forwarded over the fronthaul link. The maximum tolerable delay per device is denoted by $T_k$.

\subsection{Mathematical Formulation}
By jointly optimizing the beamforming vectors, the transmit powers, the allocated rates, and the computation allocation, we aim at maximizing the sum log-rate of all network devices, so as to strike a trade-off for balancing throughput with fairness. The optimization problem considered in this paper can then be written as:
\begin{subequations}\label{eq:Opt1}
\begin{align}
	\underset{\bm{q},\bm{p},\bm{r},\bm{f}}{\text{max}}\quad &\sideset{}{_{k\in\mathcal{K}}}\sum\log(r_k)  \tag{\ref{eq:Opt1}} \\
	\text{s.t.} \quad\,\,\;\; & \eqref{eq:maxcompcap},\eqref{eq:dir_rel_decision}, \eqref{eq:rmn1}-\eqref{eq:rmn2},\eqref{eq:rnsc}, \eqref{eq:front}, \eqref{eq:psmax},\hspace*{-.2cm} \nonumber\\
	&P_{b}^\mathrm{cc}(\bm{q}) \leq P_{b}^{\text{max}}, &\forall b \in \mathcal{B}, \label{eq:powercc1}\\
	&P_{e}^\mathrm{ec}(\bm{q},\bm{f}) \leq P_{e}^{\text{max}},  &\forall e \in \mathcal{E}, \label{eq:powerec1}\\
	&{F_k}/{f_{k}} + {D_k}/{r_k} \leq T_k - \Lambda_{k},& \forall k \in \mathcal{K}, \label{eq:delay1}
\end{align}
\end{subequations}
with beamforming vectors $\bm{q}$, transmit powers $\bm{p}$, allocated rates $\bm{r}$, and computation allocation $\bm{f}$, where $\bm{f} = [f_1,\cdots,f_K]^T$. The solution set to problem \eqref{eq:Opt1} is limited by the linear maximum computation capacity constraint per platform \eqref{eq:maxcompcap}, the mixed-integer direct/relayed-link selection constraint for weak devices \eqref{eq:dir_rel_decision}, the non-convex, fractional achievable rate per strong device \eqref{eq:rm}, serving weak device in time slot $1$ via direct link \eqref{eq:rmn1} or relayed link \eqref{eq:rnn1}, and serving weak device in time slot $2$ \eqref{eq:rmn2}, the mixed-integer total weak device rate \eqref{eq:rnsc} constraint, the non-convex, mixed-integer fronthaul capacity constraint \eqref{eq:front}, the linear power constraint of each strong device \eqref{eq:psmax}, the convex power constraint of the CC \eqref{eq:powercc1} and of each EC \eqref{eq:powerec1}, and the non-convex, mixed-integer delay constraint \eqref{eq:delay1}. Note that the objective function in \eqref{eq:Opt1} ensures a trade-off between throughput and fairness. Due to the logarithm function, poorly-serviced devices are given priority when allocating resources, whereas high rate-devices continue to receive resources but at a diminishing rate.
Problem \eqref{eq:Opt1} constitutes a non-convex, mixed-integer optimization problem, which is difficult to solve in general. To tackle problem \eqref{eq:Opt1}, this paper next proposes a sequence of reformulation steps yielding a numerically practical algorithm for radio and computational resource management in the considered network.
\section{Proposed Algorithm}\label{sec:alg}
To solve the intricate optimization problem \eqref{eq:Opt1}, the paper uses a sequence of well-chosen techniques for tackling the mixed-integer and non-convex nature of the objective and constraints set. First, the weak device link selection \eqref{eq:dir_rel_decision} and rate expression \eqref{eq:rnsc}, fronthaul capacity \eqref{eq:front}, and fronthaul delay \eqref{eq:Lambdak} constraints are tackled using a heuristic, yet reasonable, weighted $\ell_1$-norm approximation, followed by an SCA and worst-case delay approximation method. Subsequently, FP, i.e., the quadratic transform in this case, is utilized to treat the achievable rate constraints \eqref{eq:rmn1}-\eqref{eq:rmn2} by means of tactfully updating specific auxiliary variables.
Finally, the paper presents a distributed implementation framework for the proposed algorithm, which only relies on a reasonable coordination among the CC and EC platforms.

\subsection{Direct/Relay-Link Selection}
Both constraints \eqref{eq:dir_rel_decision} and \eqref{eq:rnsc}, i.e., the direct/relay-link selection and weak device rate constraints, are functions of the $\ell_0$-norm of the beamforming vectors. Since the $\ell_0$-norm returns a binary value, such constraints become a major hurdle in solving the optimization problem \eqref{eq:Opt1} efficiently.
Using a weighted $\ell_1$-norm relaxation, e.g., see \cite{DaiY14}, we replace such $\ell_0$-norms using weighted $\ell_1$-norm formulations. That is, we approximate $\norm[\big]{\norm[\big]{\bm{{q}}_{w}^{(\Upsilon)}}_{2}^{2}}_0 = \beta_{w}^{(\Upsilon)}\norm[\big]{\norm[\big]{\bm{{q}}_{w}^{(\Upsilon)}}_{2}^{2}}_1$, where $\Upsilon\in\{d,r\}$ denotes either the direct or relay link.
The weights $\beta_{w}^{(\Upsilon)}$ are defined by 
\begin{align}
\beta_{w}^{(\Upsilon)} = \big(\delta+\norm[\big]{\bm{\tilde{q}}_{w}^{(\Upsilon)}}_{2}^{2}\big)^{-1}, \label{eq:updatebeta1}
\end{align}
where $\bm{\tilde{q}}_{w}^{(\Upsilon)}$ denotes the previous iteration's beamforming vector and $\delta>1$.
Applied to \eqref{eq:dir_rel_decision} and \eqref{eq:rnsc}, the respective weighted $\ell_1$-norm relaxed constraints become:
\begin{align}
r_w - \beta_{w}^{(d)}\norm[\big]{\norm[\big]{\bm{{q}}_{w}^{(d)}}_{2}^{2}}_1 r_{w}^{(d)}\hspace*{2.2cm}& \nonumber\\
- \beta_{w}^{(r)}\norm[\big]{\norm[\big]{\bm{{q}}_{w}^{(r)}}_{2}^{2}}_1 \text{min}\left(r_{w}^{(r)},r_{w}^{(2)}\right) &\leq 0,\label{eq:rnsc2}\\
\beta_{w}^{(d)}\norm[\big]{\norm[\big]{\bm{{q}}_{w}^{(d)}}_{2}^{2}}_1 + \beta_{w}^{(r)}\norm[\big]{\norm[\big]{\bm{{q}}_{w}^{(r)}}_{2}^{2}}_1 &\leq 1.\label{eq:dir_rel_decision2}
\end{align}
With such formulation at hand, the mixed-integer nature of the link selection constraints is relaxed into a smooth continuous function. The nature of \eqref{eq:dir_rel_decision2} has the following link selection interpretation: If a low power level is assigned to device $w$'s direct or relay link, the corresponding weight increases. As the addition of both $\ell_1$-norm relaxed formulations cannot exceed $1$, the link with low allocated power is eventually shut down, and the stronger link remains active.

To further tackle \eqref{eq:rnsc2}, we introduce the auxiliary variable $\bm{z}=\big[z_{1}^{(d)},\cdots,z_{K}^{(d)},z_{1}^{(r)},\cdots,z_{K}^{(r)}\big]^T$, where only the values at the weak users $k\in\mathcal{K}^\text{weak}$ are non-zero. Further, we extend the rate variable $\bm{r}$ by the auxiliary variable $r_{w}^\text{aux}$, $\forall w\in\mathcal{K}^\text{weak}$, which captures the minimum achievable rate of device $w$'s message on the relay link. One can, therefore, formulate the weak device auxiliary constraints for all $w \in \mathcal{K}^\text{weak}$ as:
\begin{align}
\beta_{w}^{(d)}\norm[\big]{\norm[\big]{\bm{{q}}_{w}^{(d)}}_{2}^{2}}_1 &\leq z_{w}^{(d)},\label{eq:betand}\\
\beta_{w}^{(r)}\norm[\big]{\norm[\big]{\bm{{q}}_{w}^{(r)}}_{2}^{2}}_1 &\leq z_{w}^{(r)}.\label{eq:betanr}
\end{align}
The achievable rate of device $w$'s message on the relay link is given as
\begin{align}
\text{min}\left(r_{w}^{(r)},r_{w}^{(2)}\right) &\geq r_{w}^\text{aux}.\label{eq:rmn12}
\end{align}
Hence, constraints \eqref{eq:rnsc2} and \eqref{eq:dir_rel_decision2} can be rewritten as:
\begin{align}
r_w - z_{w}^{(d)} r_{w}^{(d)} - z_{w}^{(r)} r_{w}^\text{aux}&\leq 0,\label{eq:rnsc3}\\
z_{w}^{(d)} + z_{w}^{(r)} &\leq 1.\label{eq:dir_rel_decision3}
\end{align}
While constraints \eqref{eq:betand}, \eqref{eq:betanr}, \eqref{eq:rmn12}, and \eqref{eq:dir_rel_decision3} appear in convex form, \eqref{eq:rnsc3} is of bilinear nature and couples the selection-auxiliary variable and the rates. To tackle the issue of bilinearity in \eqref{eq:rnsc3} and convexifying the corresponding constraint using SCA, consider the following lemma.
\begin{lemma}\label{lma1}
The convex SCA-applied reformulation of \eqref{eq:rnsc3} is given by the following constraint for all weak devices $w\in\mathcal{K}^\text{weak}$
\begin{align}
	&4 r_w + ({z}_{w}^{(d)} - {r}_{w}^{(d)})^2 + (z_{w}^{(r)} - r_{w}^\text{aux})^2 \label{eq:rnsc5}\\
	&- (\tilde{z}_{w}^{(d)} + \tilde{r}_{w}^{(d)})^2 - (\tilde{z}_{w}^{(r)} + \tilde{r}_{w}^\text{aux})^2 \nonumber\\
	&- 2 (\tilde{z}_{w}^{(d)} + \tilde{r}_{w}^{(d)})\big[(z_w^{(d)}-\tilde{z}_w^{(d)})+({r}_{w}^{(d)}-\tilde{r}_{w}^{(d)}) \big]\nonumber\\
	&- 2 (\tilde{z}_{w}^{(r)} + \tilde{r}_{w}^\text{aux}) \big[({z}_{w}^{(r)}-\tilde{z}_{w}^{(r)})+({r}_{w}^\text{aux}-\tilde{r}_{w}^\text{aux})\big]\leq 0,\nonumber
\end{align}
where $\tilde{z}_{w}^{(d)}$, $\tilde{z}_{w}^{(r)}$, $\tilde{r}_{w}^{(d)}$, and $\tilde{r}_{w}^\text{aux}$ are optimized variables from the previous iteration.
\end{lemma}
\begin{proof}
The proof of the above lemma can be found in Appendix~\ref{app1}.
\end{proof}
Under all previously described reformulations, the direct/relay-link selection constraints \eqref{eq:betand}-\eqref{eq:rmn12}, \eqref{eq:dir_rel_decision3}, and \eqref{eq:rnsc5} now appear in convex forms, which are eventually used in deriving the overall proposed algorithm.
\subsection{Fronthaul Capacity and Delay}
Firstly, the fronthaul capacity constraint \eqref{eq:front} is tackled using an SCA approach, i.e., using similar steps to the previous lemma. 
More specifically, we let $z_{b,k}$, $\forall b \in \mathcal{B}$, $\forall k\in\mathcal{K}$, be an auxiliary variable. Using weighted $\ell_1$-norm relaxation, and applying SCA, we obtain the following set of constraints for all $b \in \mathcal{B}$ and $k\in\mathcal{K}$:
\begin{align}
\beta_{b,k}\big(\norm[\big]{\norm[\big]{\bm{{u}}_{b,k}^{(d)}}_{2}^{2}}_1+\norm[\big]{\norm[\big]{\bm{{u}}_{b,k}^{(r)}}_{2}^{2}}_1\big) &\leq z_{b,k}, \label{eq:zbk}\\
\sum_{k\in\mathcal{K}}\big( \left(z_{b,k}+r_{k}\right)^2 - 2 \left(\tilde{z}_{b,k}-\tilde{r}_{k}\right)\left(z_{b,k}-r_{k}\right)& \nonumber\\[-.35cm]
\qquad\qquad\qquad + \left(\tilde{z}_{b,k}-\tilde{r}_{k}\right)^2\big) &\leq 4 R_b^\text{max}, \label{eq:front3}
\end{align}
with weights 
\begin{align}
\beta_{b,k} = \big(\delta + \norm[\big]{\bm{\tilde{u}}_{b,k}^{(d)}}_{2}^{2} + \norm[\big]{\bm{\tilde{u}}_{b,k}^{(r)}}_{2}^{2}\big)^{-1},\label{eq:updatebeta2}
\end{align}
and where $\bm{\tilde{u}}_{b,k}^{(d)}$, $\bm{\tilde{u}}_{b,k}^{(r)}$, and $\tilde{z}_{b,k}$ are the beamforming vectors and auxiliary variables from the previous iteration. With a similar reasoning to the link selection constraint, using the weighted $\ell_1$-norm approach results in BS-to-device associations being deactivated during the algorithm execution at low levels of associated powers. 

Secondly, the fronthaul delay \eqref{eq:Lambdak} is tackled based on the previous iteration's BS-to-device clustering. That is, knowing that a link can only be deactivated from one iteration to another (not activated), the worst-case fronthaul delay of a device $k$ is, therefore, given by:
\begin{align}
\Lambda_k = \underset{b\in\mathcal{B}_k}{\text{max}}\Bigg\{ \frac{1}{R_b^\text{max}} \sum_{i\in\mathcal{K}_c} \bigg( \norm[\big]{\norm[\big]{\bm{\tilde{u}}_{b,i}^{(d)}}_{2}^{2}}_0 +\norm[\big]{\norm[\big]{\bm{\tilde{u}}_{b,i}^{(r)}}_{2}^{2}}_0 \bigg)  D_{i} \Bigg\}.\label{eq:Lambdak2}
\end{align}
Formulation \eqref{eq:Lambdak2}, in fact, only depends on the previous iteration's optimization variables, and thus guarantees an upper bound on the fronthaul delay at each iteration. Now that this step is solved, constraint \eqref{eq:zbk}, constraint \eqref{eq:front3}, and the worst-case fronthaul delay \eqref{eq:Lambdak2} are of convex form. At this stage, the remaining non-convexity of the original problem \eqref{eq:Opt1} stems from the achievable rate constraints \eqref{eq:rmn1}-\eqref{eq:rmn2}.
\subsection{Achievable Rate Constraints}
The remaining hurdles against solving the original optimization problem \eqref{eq:Opt1} are the highly-coupled, fractional, and non-convex achievable rate formulations \eqref{eq:rmn1}-\eqref{eq:rmn2}. We next tackle such constraints using FP, by introducing the auxiliary variables $\gamma_{w}^{(r)}$, $\gamma_{s}$, $\gamma_{w}^{(d)}$, and $\gamma_{w}^{(2)}$, $\forall s\in\mathcal{K}^\text{str}$ and $ w\in\mathcal{K}^\text{weak}$. We then obtain the following constraints:
\begin{align}
r_{w}^{(r)} &\leq \nu W\log_2(1+\gamma_{w}^{(r)}), \label{eq:rmn1_convex}\\
r_{s} &\leq \nu W\log_2(1+\gamma_{s}),\label{eq:rm_convex}\\
r_{w}^{(d)} &\leq \nu W\log_2(1+\gamma_{w}^{(d)}),\label{eq:rnn1_convex}\\
r_{w}^{(2)} &\leq (1-\nu) W\log_2(1+\gamma_{w}^{(2)}), \label{eq:rmn2_convex}
\end{align}\vspace*{-.7cm}
\begin{align}
&\gamma_{w}^{(r)} \leq \Gamma_{w}^{(r)};\,\gamma_{s} \leq \Gamma_{s};\,\gamma_{w}^{(d)} \leq \Gamma_{w}^{(d)};\,\gamma_{w}^{(2)} \leq \Gamma_{w}^{(2)}. \label{eq:SINR_all} 
\end{align}%
Note that the above constraints within \eqref{eq:SINR_all} have a non-linear form. Their fractional type, however, allows using FP techniques, namely the quadratic transform, so as to reach an easier solution to the original optimization problem.

\begin{lemma}\label{lma3}
Applying the quadratic transform to \eqref{eq:SINR_all}, we obtain the convex functions $g_{w}^{(r)}(\bm{q})$, $g_{w}^{(d)}(\bm{q})$, $g_{w}^{(2)}(\bm{p})$, and $g_{s}(\bm{q})$, which can be mathematically written as follows:
\begin{align}
	g_{w}^{(r)}(\bm{q}) &= \gamma_{w}^{(r)} - 2 \text{Re}\big\{ (a_{w}^{(r)})^H \bm{h}_{d_w}^{H} \bm{q}_{w}^{(r)} \big\} \nonumber\\
	&\hspace*{2cm} + |a_{w}^{(r)}|^2 I_{d_w}(\mathcal{K}^\text{str},\mathcal{K}^\text{weak}\backslash\{w\}),\label{eq:gsw1}\\
	g_{w}^{(d)}(\bm{q}) &= \gamma_{w}^{(d)} - 2 \text{Re}\big\{ (a_{w}^{(d)})^H \bm{h}_{w}^{H} \bm{q}_{w}^{(d)} \big\} \nonumber\\
	&\hspace*{2cm} + |a_{w}^{(d)}|^2 I_w(\mathcal{K}^\text{str},\mathcal{K}^\text{weak}\backslash\{w\}),\label{eq:gww1}\\
	g_{w}^{(2)}(\bm{p}) &= \gamma_{w}^{(2)} - 2 a_{w}^{(2)} \sqrt{|g_{d_w,w}|^2 {p_{d_w}}} \nonumber\\
	&\hspace*{-.2cm} + (a_{w}^{(2)})^2 \Bigg[ \sigma^2 + {\sideset{}{_{w'\in\mathcal{K}^\text{weak}\backslash\{w\}}}\sum |g_{d_{w'},w}|^2 {p_{d_{w'}}} } \Bigg],\label{eq:gww2s}
\end{align}
\begin{align}
	&g_{s}(\bm{q}) = \nonumber\\
	&\begin{cases}
		\gamma_{s} \hspace*{-.1cm}- 2 \text{Re}\big\{\hspace*{-.05cm}a_{s}^H \bm{h}_{s}^{H} \bm{q}_{s}^{(d)}\hspace*{-.05cm}\big\} + |a_{s}|^2 I_s(\mathcal{K}^\text{str}\backslash\{s\},\mathcal{K}^\text{weak}\backslash\{d_s\}), \\
		\hspace*{5.8cm}\text{if } x_{d_s} \text{ is decoded},\\
		\gamma_{s} \hspace*{-.1cm}- 2 \text{Re}\big\{\hspace*{-.05cm}a_{s}^H \bm{h}_{s}^{H} \bm{q}_{s}^{(d)}\hspace*{-.05cm}\big\} + |a_{s}|^2  I_s(\mathcal{K}^\text{str}\backslash\{s\},\mathcal{K}^\text{weak}), \\
		\hspace*{5.8cm}\text{otherwise}.
	\end{cases}\label{eq:gss}
\end{align}
The corresponding optimal auxiliary variables, namely, $a_{w}^{(r)}$, $a_{w}^{(d)}$, $a_{w}^{(2)}$, and $a_{s}$, are combined in the vector $\bm{a}$ and updated iteratively using the following update equations:
{\begin{align}
		(a_{w}^{(r)})^* &= 2 \text{Re}\big\{\bm{h}_{d_w}^{H} \bm{\tilde{q}}_{w}^{(r)} \big\} /  \tilde{I}_{d_w}(\mathcal{K}^\text{str},\mathcal{K}^\text{weak}\backslash\{w\}), \label{eq:amn1}\\
		(a_{w}^{(d)})^* &= 2 \text{Re}\big\{\bm{h}_{w}^{H} \bm{\tilde{q}}_{w}^{(d)} \big\} / \tilde{I}_w(\mathcal{K}^\text{str},\mathcal{K}^\text{weak}\backslash\{w\}), \label{eq:ann1}\\
		(a_{w}^{(2)})^* &= \frac{\sqrt{|g_{d_w,w}|^2 {\tilde{p}_{d_w}}}}{\sigma^2 + {\sideset{}{_{w'\in\mathcal{K}^\text{weak}\backslash\{w\}}}\sum |g_{d_{w'},w}|^2 {\tilde{p}_{d_{w'}}} }} \label{eq:ann2},\\
		a_{s}^* &= \begin{cases}
			\frac{2 \text{Re}\{\bm{h}_{s}^{H} \bm{\tilde{q}}_{s}^{(d)}\}}{\tilde{I}_s(\mathcal{K}^\text{str}\backslash\{s\},\mathcal{K}^\text{weak}\backslash\{d_s\})}, &\text{if } x_{d_s} \text{ is decoded},\\[.1cm]
			\frac{2 \text{Re}\{\bm{h}_{s}^{H} \bm{\tilde{q}}_{s}^{(d)}\} }{\tilde{I}_s(\mathcal{K}^\text{str}\backslash\{s\},\mathcal{K}^\text{weak})}, &\text{otherwise}.
		\end{cases}\label{eq:amm}
\end{align}}%
Note that in \eqref{eq:amn1}, \eqref{eq:ann1}, and \eqref{eq:amm}, the function $\tilde{I}_i(\mathcal{X},\mathcal{Y})$ is adopted from \eqref{eq:I_i}, where the beamforming vectors are replaced by the previous iteration's optimal solution, i.e., $\bm{\tilde{q}}$.
\end{lemma}
\begin{proof}
The detailed steps of the proof can be found in Appendix~\ref{app2}.
\end{proof}
Using the above lemma, the achievable rate constraints \eqref{eq:rmn1_convex}-\eqref{eq:rmn2_convex} and the auxiliary SINR-related constraints \eqref{eq:gsw1}-\eqref{eq:gss} are, in fact, now in convex forms. All of the above make problem \eqref{eq:Opt1}'s intricacies tackled, as they result in a convex optimization problem, the solution of which is computed in an iterative fashion by updating weights and auxiliary variables after each iteration, as illustrated next.

\subsection{Centralized Resource Management (CRM) Protocol}\label{ssec:prob}
\subsubsection{Algorithmic Description}
With all of the above considerations at hand, we can reformulate the original optimization problem as follows:
\begin{subequations}\label{eq:Opt2}
\begin{align}
	\underset{\bm{q},\bm{p},\bm{r},\bm{f},\bm{z},\boldsymbol{\gamma}}{\text{max}}\; &\sideset{}{_{k\in\mathcal{K}}}\sum\log(r_k)  \tag{\ref{eq:Opt2}} \\
	\text{s.t.} \quad\,\;\; & \eqref{eq:maxcompcap}, \eqref{eq:psmax}, \eqref{eq:powercc1}-\eqref{eq:delay1}, \eqref{eq:betand}-\eqref{eq:rmn12}, \eqref{eq:dir_rel_decision3}, \hspace*{-2.99cm}\nonumber\\
	& \eqref{eq:rnsc5}-\eqref{eq:front3}, \eqref{eq:rmn1_convex}-\eqref{eq:rmn2_convex},\hspace*{0.1cm}\nonumber\\
	&g_{w}^{(r)}(\bm{q}) \leq 0, &\forall w\in\mathcal{K}^\text{weak},\label{eq:gsw1_}\\
	&g_{w}^{(d)}(\bm{q}) \leq 0, &\forall w\in\mathcal{K}^\text{weak},\label{eq:gww1_}\\
	&g_{w}^{(2)}(\bm{p}) \leq 0, &\forall w\in\mathcal{K}^\text{weak},\label{eq:gww2s_}\\
	&g_{s}(\bm{q}) \leq 0, &\forall s\in\mathcal{K}^\text{str}.\label{eq:gss_}
\end{align}
\end{subequations}
\begin{algorithm}[t]
\caption{Centralized Resource Management}
\begin{algorithmic}[1]
	\STATE Initialize $\bm{q}$ and $\bm{p}$ to feasible values\\ \vspace*{-.05cm}
	\textbf{Repeat:} until convergence \vspace*{-.05cm}
	\STATE Update weights $\boldsymbol{\beta}$ using \eqref{eq:updatebeta1}, \eqref{eq:updatebeta2} \vspace*{-.05cm}
	\STATE Update $\bm{\tilde{q}}$, $\bm{\tilde{r}}$, $\bm{\tilde{z}}$, and $\bm{\tilde{p}}$ using the previous solution \vspace*{-.05cm}
	\STATE Update auxiliary variables $\bm{a}$ using \eqref{eq:amn1}-\eqref{eq:amm} \vspace*{-.05cm}
	\STATE Solve convex optimization problem \eqref{eq:Opt2}\vspace*{-.05cm}
	\STATE \textbf{End}  \vspace*{-.05cm}
\end{algorithmic}
\label{alg:crm}
\end{algorithm}%
Problem \eqref{eq:Opt2} jointly optimizes the beamforming vectors $\bm{q}$, the strong device transmit powers $\bm{p}$, the allocated rates $\bm{r}$, the computation allocations $\bm{f}$, and the auxiliary variables $\bm{z}$ and $\boldsymbol{\gamma}$. As described earlier, problem \eqref{eq:Opt2} is subject to maximum computation capacity \eqref{eq:maxcompcap}, strong device transmit power \eqref{eq:psmax}, BS transmit power \eqref{eq:powercc1}, EC power \eqref{eq:powerec1}, and per-device delay \eqref{eq:delay1} constraints. Further, the direct/relay-link selection constraints \eqref{eq:betand}-\eqref{eq:rmn12}, \eqref{eq:dir_rel_decision3}, \eqref{eq:rnsc5}, the fronthaul capacity constraints \eqref{eq:zbk}, \eqref{eq:front3}, the achievable rate constraints \eqref{eq:rmn1_convex}-\eqref{eq:rmn2_convex}, and the quadratic transform SINR constraints \eqref{eq:gsw1_}-\eqref{eq:gss_} define the convex constraint set of problem \eqref{eq:Opt2}.

Hence, \eqref{eq:Opt2} is a convex problem, as it maximizes a concave objective function subject to convex constraints. It can, therefore, be solved efficiently, e.g., using CVX \cite{cvx}. The adopted centralized resource management (CRM) algorithm then runs as follows. Initially, the beamforming vectors $\bm{q}$ and power allocations $\bm{p}$ are initialized to feasible values. In this work, $\bm{q}$ is initialized randomly, and $\bm{p}$ is initially set to the maximum nominal value. The following steps are then repeated until convergence, i.e., until the sum-logarithmic rate increase between iterations falls below a certain threshold $\omega$. The CRM updates the weights $\beta_w^{(d)}$, $\beta_w^{(r)}$, and $\beta_{b,k}$, the feasible fixed values $\bm{\tilde{q}}$, $\bm{\tilde{r}}$, $\bm{\tilde{z}}$, and $\bm{\tilde{p}}$ according to the last iteration's solution, and the optimal auxiliary variables $\bm{a}$ according to the update equations \eqref{eq:amn1}-\eqref{eq:amm}. Note that in the initial iteration, $\bm{\tilde{q}}$ and $\bm{\tilde{p}}$ are taken from the initialization, $\bm{r}$ is set by replacing the inequalities by equalities in \eqref{eq:rmn1}-\eqref{eq:rmn2} and \eqref{eq:rmn12}, and $\bm{z}$ by an all-ones vector. Subsequently, the CRM solves the convex optimization problem \eqref{eq:Opt2} using CVX \cite{cvx}. The above steps are summarized in the description of Algorithm~\ref{alg:crm}.
\subsubsection{Computational Complexity}\label{ssec:compcompcrm}
Algorithm~\ref{alg:crm}'s computational complexity is a function of the following two factors: (i) The complexity of problem \eqref{eq:Opt2} and the used solver; (ii) The convergence rate of the algorithm, i.e., the maximum number of iterations. The convex problem \eqref{eq:Opt2}, solvable using the interior-point method, has a total number of variables $\eta = 6K+3K^\text{weak}+K_cB+BL_c(K_c+K_c^\text{weak})+\sideset{}{_{e\in\mathcal{E}}}\sum L_e(K_e+K^\text{weak}_e)$, where $K^\text{weak}=|\mathcal{K}^\text{weak}|$, $K_c^\text{weak} = |\{\mathcal{K}_c\cap \mathcal{K}^\text{weak}\}|$ and $K_e^\text{weak} = |\{\mathcal{K}_e\cap \mathcal{K}^\text{weak}\}|$. By further letting the worst-case number of iterations be $\Omega^\text{max}$, the CRM's upper-bound computational complexity is $\mathcal{O}(\Omega^\text{max}(\eta)^{3.5})$ \cite{Lobo1998ApplicationsOS}.

Such a centralized approach, however, may require a signal level coordination among the CC and ECs, which is challenging in a reasonably sized network. Henceforth, the next subsection describes how the proposed approach can instead be implemented in a distributed fashion across the CC and multiple ECs, which makes the proposed algorithm amenable for practical implementation.
\subsection{Distributed Resource Management (DRM) Protocol}\label{ssec:DRM_prob}
\subsubsection{Algorithmic Description}
As discussed above, CRM implementation requires signal-level coordination among the computing platforms (i.e., the CC and ECs). Additionally, large-scale hybrid CC/MEC networks host a huge number of devices, BSs, and UAVs, which adds to the computational burden of to the CRM algorithm. Therefore, this paper proposes an alternative, feasible, distributed resource management (DRM) protocol, which can be implemented in a distributed fashion across the CC and the ECs. Such distributed operation relaxes the need of the computing platforms to be fully connected, and rather relies on reasonable message exchange among the computing platforms. More specifically, considering the convex optimization problem \eqref{eq:Opt2}, we note that the fronthaul-related constraints \eqref{eq:zbk}, \eqref{eq:front3} only affect the CC network part. Link selection \eqref{eq:betand}-\eqref{eq:rmn12}, \eqref{eq:dir_rel_decision3}, \eqref{eq:rnsc5}, computation aspects \eqref{eq:maxcompcap}, power constraints \eqref{eq:psmax}, \eqref{eq:powercc1}, \eqref{eq:powerec1}, achievable rates \eqref{eq:rmn1_convex}-\eqref{eq:rmn2_convex}, and delay formulations \eqref{eq:delay1} are also dependent on each computing platform, respectively. Constraints \eqref{eq:gsw1_}-\eqref{eq:gss_}, however, are not on a per-computing platform basis, as they depend on the following terms:
\begin{align}
&\sideset{}{_{s'\in\mathcal{K}^\text{str}}}\sum |\bm{h}_{s}^{H} \bm{q}_{s'}^{(d)}|^2; \sideset{}{_{w'\in\mathcal{K}^\text{weak}}}\sum |\bm{h}_{s}^{H} (\bm{q}_{w'}^{(d)}+\bm{q}_{w'}^{(r)})|^2;\label{eq:centr_term1}\\
&\sideset{}{_{s\in\mathcal{K}^\text{str}}}\sum |\bm{h}_{w}^{H} \bm{q}_{s}^{(d)}|^2; \sideset{}{_{w'\in\mathcal{K}^\text{weak}}}\sum |\bm{h}_{w}^{H} (\bm{q}_{w'}^{(d)}+\bm{q}_{w'}^{(r)})|^2;\label{eq:centr_term2}\\
&\sideset{}{_{w'\in\mathcal{K}^\text{weak}\backslash\{w\}}}\sum |g_{d_{w'},w}|^2 {p_{d_{w'}}}.\label{eq:centr_term3}
\end{align}
The above terms are in fact used in updating the auxiliary variables $\big(a_{w}^{(r)}\big)^*$, $\big(a_{w}^{(d)}\big)^*$, $\big(a_{w}^{(2)}\big)^*$, and $a_{s}^*$, and require knowledge about the variables of the whole network. We next present how to decouple the decisions among the computing platforms. First note that the CC only has knowledge of
\begin{align}
&|\bm{h}_{s}^{H} \bm{q}_{i}|^2, &&\forall s\in\{\mathcal{K}_c\cap\mathcal{K}^\text{str}\}, \forall i\in\mathcal{K}_c,\\
&|\bm{h}_{w}^{H} \bm{q}_{i}|^2, &&\forall w\in\{\mathcal{K}_c\cap\mathcal{K}^\text{weak}\}, \forall i\in\mathcal{K}_c,\\
&|g_{s',w'}|^2 {p_{s'}}, \hspace*{-.2cm}&&\forall s'\in\{\mathcal{K}_c\cap\mathcal{K}^\text{str}\}, \forall w'\in\{\mathcal{K}_c\cap\mathcal{K}^\text{weak}\}.
\end{align}
Each EC $e$, on the other hand, has access to
\begin{align}
&|\bm{h}_{s}^{H} \bm{q}_{i}|^2, &&\forall s\in\{\mathcal{K}_e\cap\mathcal{K}^\text{str}\}, \forall i\in\mathcal{K}_e,\\
&|\bm{h}_{w}^{H} \bm{q}_{i}|^2, &&\forall w\in\{\mathcal{K}_e\cap\mathcal{K}^\text{weak}\}, \forall i\in\mathcal{K}_e,\\
&|g_{s',w'}|^2 {p_{s'}}, \hspace*{-.2cm}&&\forall s'\in\{\mathcal{K}_e\cap\mathcal{K}^\text{str}\}, \forall w'\in\{\mathcal{K}_e\cap\mathcal{K}^\text{weak}\}.
\end{align}
To be able to compute \eqref{eq:centr_term1}-\eqref{eq:centr_term3}, i.e., to enable the distributed computations of the interference at the CC and EC platforms, respectively, the CC needs to communicate the following locally known terms to all ECs, i.e., to each $e\in\mathcal{E}$:
\begin{align}
&|\bm{h}_{s}^{H} \bm{q}_{i}|^2, &&\forall s\in\{\mathcal{K}_e\cap\mathcal{K}^\text{str}\}, \forall i\in\mathcal{K}_c,\label{eq:DRM_CC_exchange_1}\\
&|\bm{h}_{w}^{H} \bm{q}_{i}|^2, &&\forall w\in\{\mathcal{K}_e\cap\mathcal{K}^\text{weak}\}, \forall i\in\mathcal{K}_c,\label{eq:DRM_CC_exchange_2}\\
&|g_{s',w'}|^2 {p_{s'}}, \hspace*{-.2cm}&&\forall s'\in\{\mathcal{K}_c\cap\mathcal{K}^\text{str}\}, \forall w'\in\{\mathcal{K}_e\cap\mathcal{K}^\text{weak}\}.\label{eq:DRM_CC_exchange_3}
\end{align}
Similarly, each EC $e$ needs to send out its locally known terms to all other computing platforms:
\begin{align}
&|\bm{h}_{s}^{H} \bm{q}_{i}|^2, &&\forall s\in\mathcal{K}^\text{str}\backslash\{\mathcal{K}_e\}, \forall i\in\mathcal{K}_e,\label{eq:DRM_EC_exchange_1}\\
&|\bm{h}_{w}^{H} \bm{q}_{i}|^2, &&\forall w\in\mathcal{K}^\text{weak}\backslash\{\mathcal{K}_e\}, \forall i\in\mathcal{K}_e,\label{eq:DRM_EC_exchange_2}\\
&|g_{s',w'}|^2 {p_{s'}}, \hspace*{-.2cm}&&\forall s'\in\{\mathcal{K}_e\cap\mathcal{K}^\text{str}\}, \forall w'\in\mathcal{K}^\text{weak}\backslash\{\mathcal{K}_e\}.\label{eq:DRM_EC_exchange_3}
\end{align}
Based on the above information exchange, the DRM procedure then applies similar CRM steps, except adding an additional step, which now relies on exchanging the interference terms \eqref{eq:DRM_CC_exchange_1}-\eqref{eq:DRM_EC_exchange_3} among the computing platforms (i.e., CC and ECs), e.g., using an independent control module. Fig.~\ref{fig:distributed_exchange_protocol} illustrates the information interaction performed by the DRM in a network consisting of one CC and two ECs. That is, the procedure involves three sub-steps, namely, $1$, the CC distributes \eqref{eq:DRM_CC_exchange_1}-\eqref{eq:DRM_CC_exchange_3} to the ECs, as well as $2$ and $3$, where each EC distributes \eqref{eq:DRM_EC_exchange_1}-\eqref{eq:DRM_EC_exchange_3} to the CC and the other EC, respectively. The steps of the overall DRM method are described in detail in Algorithm~\ref{alg:drm}.
\begin{figure}[t]
\centering
\includegraphics[width=.99\linewidth]{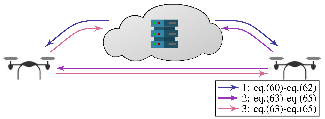}
\vspace*{-.25cm}
\caption{Distributed information exchange protocol.}
\label{fig:distributed_exchange_protocol}
\end{figure}
\begin{algorithm}[b]
\caption{Distributed Resource Management}
\begin{algorithmic}[1]
	\STATE Initialize $\bm{q}$ and $\bm{p}$ to feasible values\\ \vspace*{-.05cm}
	\textbf{Repeat:} until convergence \vspace*{-.05cm}
	\STATE Update weights $\boldsymbol{\beta}$ using \eqref{eq:updatebeta1}, \eqref{eq:updatebeta2} \vspace*{-.05cm}
	\STATE Update $\bm{\tilde{q}}$, $\bm{\tilde{r}}$, $\bm{\tilde{z}}$, and $\bm{\tilde{p}}$ using the previous solution \vspace*{-.05cm}
	\STATE Update auxiliary variables $\bm{a}$ using \eqref{eq:amn1}-\eqref{eq:amm} \vspace*{-.05cm}
	\STATE Solve convex optimization problem \eqref{eq:Opt2}\vspace*{-.05cm}
	\STATE Exchange \eqref{eq:DRM_CC_exchange_1}-\eqref{eq:DRM_EC_exchange_3} between platforms\vspace*{-.05cm}
	\STATE \textbf{End}  \vspace*{-.05cm}
\end{algorithmic}
\label{alg:drm}
\end{algorithm}%
\subsubsection{Computational Complexity}\label{ssec:compcomp}
The DRM implementation involves $\eta_c = (6+B)K_c+3K_c^\text{weak}+BL_c(K_c+K_c^\text{weak})$ variables at the CC, and $\eta_e = 6K_e+3K_e^\text{weak}+L_e(K_e+K_e^\text{weak})$ variables at each EC $e$. 
The DRM's upper-bound computational complexities for the CC and for each EC $e$ become, therefore, $\mathcal{O}(\Omega^\text{max}(\eta_c)^{3.5})$ and $\mathcal{O}(\Omega^\text{max}(\eta_e)^{3.5})$, respectively. For completeness, the complexities of CRM and DRM are compared in Tab.~\ref{tb:complexity}.

Clearly, using DRM substantially reduces the computational burden by distributing the decisions among platforms, as opposed to CRM. While such computational relief comes at the cost of a performance loss, the paper simulations results next highlight that DRM performs close to CRM, and outperforms the benchmarks (i.e., SDMA, NOMA) in various setups, all whilst exhibiting a valuable scalabilty in terms of runtime, especially in dense networks.
\begin{table}[t]
\vspace*{.1cm}
\renewcommand{\arraystretch}{1.05}
\centering
\begin{tabular}{c | c | c}
	\hline
	\hspace*{-.2cm}Scheme\hspace*{-.15cm} & Central cloud & Edge computer \\\hline\hline
	CRM & \multicolumn{2}{c}{\hspace*{-.9cm}$\mathcal{O}(\Omega^\text{max}(6K+3K^\text{weak}+K_cB+BL_c(K_c+K_c^\text{weak})$\hspace*{-.15cm}} \\
	& \multicolumn{2}{c}{\hspace*{2.9cm}$+\sideset{}{_{e\in\mathcal{E}}}\sum L_e(K_e+K^\text{weak}_e)^{3.5})$\hspace*{-.15cm}} \\
	\hline
	DRM & \hspace*{-.3cm}$\mathcal{O}(\Omega^\text{max}((6+B)K_c+3K_c^\text{weak}$\hspace*{-.15cm} & \hspace*{-.8cm}$\mathcal{O}(\Omega^\text{max}(6K_e+3K_e^\text{weak}$\\
	& \hspace*{.6cm}$+BL_c(K_c+K_c^\text{weak}))^{3.5})$\hspace*{-.15cm} & \hspace*{.6cm}$+L_e(K_e+K_e^\text{weak}))^{3.5})$\\
	\hline
\end{tabular}
\vspace*{-.15cm}
\caption{Computational complexities of CRM and DRM.}
\label{tb:complexity}
\vspace*{-.1cm}
\end{table}%


\section{Simulations}\label{sec:sim}
To assess the numerical performance of the proposed algorithms, we consider a hybrid CC/MEC network which enables XR use-cases, e.g., virtual reality video streaming or other XR-related computation tasks. The simulation parameters are given in Tab.~\ref{tb:simparam}. Note that, at the network edge, a total of $2E$ strong and weak devices are served. Devices are paired in groups of two, where the distance of strong and weak devices can be up to $100$ m. While the CC devices are placed inside the hexagonal cell-layout, the EC devices are placed randomly outside the hexagonal BS cells; each UAV is placed above its served strong device.
\begin{table}[t]
\renewcommand{\arraystretch}{1.1}
\centering
\begin{tabular}{l | l | l}
	\hline 
	CC-related & EC-related & Device-related\\\hline
	\hline
	$B=4$ BSs & $E=6$ UAVs & $K=30$ devices\\
	$L_c=6$ antennas & $L_e=2$ antennas & \\
	$500$ m inter-BS distance & $125$ m UAV altitude & \\
	$P_b^\text{max}=32$ dBm & $P_e^\text{max}=22$ dBm & $P_s^\text{max}=20$ dBm\\
	$f_0^\text{max}=5\cdot 10^{10}$ cycles/s & $f_e^\text{max}=10^9$ cycles/s & $F_k=10^7$ cycles\\
	$R_b^\text{max}=100$ Mbps & $o_e = 10^{-28}$, $\mu_e = 3$ & $D_k=10^4$ bits\\
	$W=10$ MHz & $\Phi_e=100$ W & $T_k=50$ ms\\\hline
\end{tabular}
\vspace*{-.15cm}
\caption{Simulation parameters.}
\label{tb:simparam}
\end{table}%

The XR control input, e.g., movement and field-of-view data, is of small data size and assumed to be transmitted in the order of $1$ ms \cite{9363888}. At the CC and at each EC, the virtual environment is modeled \cite{9363888}, the video streams are encoded \cite{10033423}, or other processing tasks are computed \cite{8764580}. Such processing tasks are in the order of $10^2$ \cite{8764580} to $10^3$ \cite{10033423,8434285} cycles/bit. Afterwards, the encoded video stream or task results are transmitted to the devices. As can be seen in Tab.~\ref{tb:simparam}, the delay requirements per device are set as $T_k = 50$ ms \cite{10033423}. The related computation and power requirements are set according to \cite{8764580,8434285}.
For illustration, we set the noise power-spectral density to $-159$ dBm/Hz and the convergence threshold to $\omega=0.1$. 
The beamforming vectors are initialized randomly, and the transmit powers are initially set to their maximum values. 

To assess the proposed algorithm's performance, we apply the CRM, which forgoes the inter-platform cooperation limits and virtually treats the CC and all ECs as one resource management entity, and also compare it to the DRM performance. We also utilize SDMA \cite{10008481,5074583} and NOMA \cite{7973146} as benchmarks to compare the proposed Co-NOMA scheme against various state-of-the-art transmission schemes. The NOMA scheme specifically utilizes the same device pairing as the Co-NOMA scheme, enabling one-layer of successive interference cancellation at each strong device only.

\subsection{Channel Model}
The considered simulations adopt three different types of channels:
$(a)$ The BS-device channels utilize the 3GPP specified pathloss $\text{PL}_{k,b}(dB) = 128.1 + 37.6 \log_{10}(\text{dist}_{k,b})$, where $\text{dist}_{k,b}$ is the distance between BS $b$ and device $k$ in km; $(b)$ The UAV-device channels utilize $\text{PL}_{k,e}(dB) = 20 \log_{10}(\text{dist}_{k,e}) + 20 \log_{10}(4\pi/c) + 20 \log_{10}(f) + \text{Pr}\{\text{LoS}\} \cdot 6 + \text{Pr}\{\text{NLoS}\} \cdot 20$, where $\text{Pr}\{\text{LoS}\}$ denotes the probability of a line of sight (LoS) connection defined as $\text{Pr}\{\text{LoS}\}=(1+9.61\cdot\text{exp}(-0.16\cdot(\theta-9.61)))^{-1}$, with elevation angle $\theta$, and $\text{Pr}\{\text{NLoS}\} = 1 - \text{Pr}\{\text{LoS}\}$, $f$ is the carrier frequency set to $5$ GHz, and $c$ is the speed of light \cite{pl}; $(c)$ The D2D channels follow the short-range outdoor channel model ITU-1411, e.g., similar to \cite{8006944}, with $5$ GHz carrier frequency, antenna height of $1.5$ m, and $2.5$ dB antenna gain. Each channel $(a)-(c)$ consists of pathloss and Rayleigh fading ($\mathcal{CN}(0,1)$), and accounts for the log-normal shadowing effect (BS-device $8$ dB, UAV-device $4$ dB, and device-device $10$ dB standard deviation).

\subsection{Impact of Blockage}
\begin{figure}[t]
\centering
\includegraphics[width=.99\linewidth]{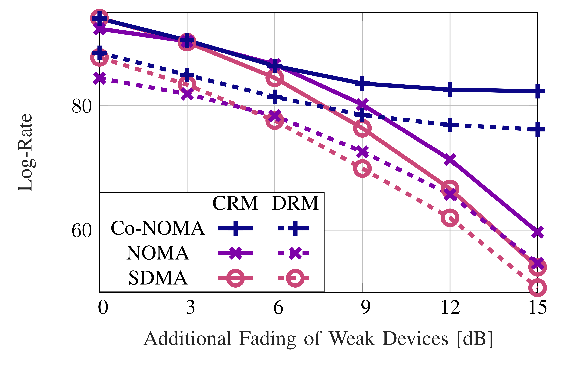}
\vspace*{-.35cm}
\caption{Log-rate over the weak devices' additional fading for all considered schemes.}
\label{res:yLogR_xAddFading_all}
\end{figure}
\begin{figure}[t]
\centering
\includegraphics[width=.99\linewidth]{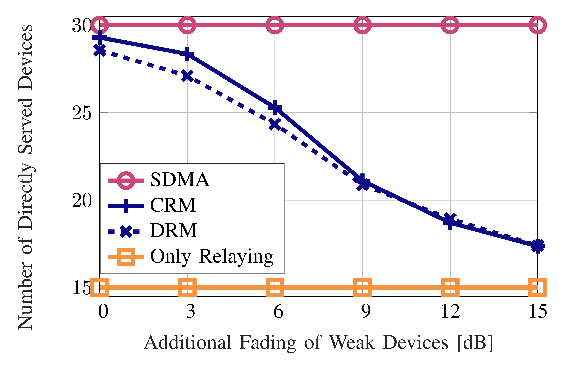}
\vspace*{-.35cm}
\caption{The number of directly served weak devices versus the additional fading.}
\label{res:yNumDirect_xTDMAfactor}
\end{figure}
An essential motivation of the Co-NOMA scheme is the resilience to adverse channel conditions or blockage events at the weak devices. Hence, in the following set of simulations, we illustrate the system performance as a function of the incremental fading experienced by each weak device, hereafter denoted by \emph{additional fading of weak devices}. Such parameter allows to capture the channel conditions of the weak devices, i.e., at high values, weak devices experience poor channel conditions. In contrast, low values indicate that weak devices experience comparable channel quality to their strong counterparts. We also note that we perform a linear search over the time-split factor $\nu$, where we search the values $\nu\in[0.4,0.5,\cdots,1]$. That is, we consider allocating $60$\%, $50$\%, $\cdots$, and $0$\% of the time to the relaying procedure. At each point of the x-axis for Co-NOMA, we then choose the factor $\nu$ that maximizes the log-rate objective.

\textbf{Log-rate performance:} 
First, in Fig.~\ref{res:yLogR_xAddFading_all}, we compare Co-NOMA, NOMA, and SDMA for both CRM and DRM, each in terms of log-rate versus the weak device additional fading. The figure first shows how the log-rate decreases as the additional fading increases. Such result is indeed expected as the RAN looses its available channel resources whenever the weak devices become subject to higher channel fading.
Fig.~\ref{res:yLogR_xAddFading_all} particularly shows how the proposed optimized scheme, Co-NOMA, achieves the highest log-rate. Interestingly, even at good channel conditions, e.g., $0$ or $3$ dB, Co-NOMA can achieve the envelope of the SDMA scheme thanks to its flexibility in link selection. That is, in this region, relaying is not effective as compared to serving all devices directly. Additionally, compared to SDMA and NOMA, which gradually decrease after $6$ dB, Co-NOMA keeps the log-rate stabilized and provides an almost constant log-rate despite the worsening of the channel conditions.
Further, Fig.~\ref{res:yLogR_xAddFading_all} shows that the distributed Co-NOMA, i.e., DRM, outperforms the distributed SDMA and NOMA in all points. It even overcomes the CRM NOMA and CRM SDMA after $9$ dB additional fading.
The figure, in fact, shows that the proposed DRM Co-NOMA remains at a constant gap from the centralized algorithm, i.e., CRM.
Further, it is worthwhile to note that the DRM achieves such a valuable performance despite its relatively faster runtime compared to the CRM and under realistic coordination constraints.
Fig.~\ref{res:yLogR_xAddFading_all}, therefore, clearly motivates the improvement of Co-NOMA compared to the benchmark schemes, thereby underlining its applicability in future 6G networks.

\textbf{Direct/relay link selection:} 
To further investigate the Co-NOMA perspective in terms of strong-to-weak device relaying, Fig.~\ref{res:yNumDirect_xTDMAfactor} depicts the number of direct served devices over the weak device additional fading.
The proposed Co-NOMA schemes are shown to operate within the bounds of serving all devices directly, i.e., SDMA, and serving all weak devices through their respective strong devices, i.e., only relaying.
Fig.~\ref{res:yNumDirect_xTDMAfactor} supports the previous results on the intelligent adaption of Co-NOMA to the channel conditions. On the left hand side of Fig.~\ref{res:yNumDirect_xTDMAfactor}, i.e., in the regimes where the strong and weak device channel conditions are relatively similar, Co-NOMA tends toward the SDMA/only direct links solution. As the channel conditions become uneven, the scheme flexibly enables more weak devices to be served using the relay link.
These results emphasize the flexibility of the proposed algorithms in terms of link selection and adapting to adverse channel conditions. Specifically for future communication networks, such flexible multiple access and resource management schemes becomes vital in many use-cases, e.g., remote diagnostics and immersive gaming.

\textbf{Time-split factor:} 
To illustrate the impact of the time split factor on the network performance, Fig.~\ref{res:yLogR_xTDMAfactor_FCP} plots the log-rate for the CRM Co-NOMA versus $\nu$ using different values of the weak device additional fading values.
Fig.~\ref{res:yLogR_xTDMAfactor_FCP} validates how higher fading values decrease the log-rate. However, the behavior with respect to the time-split factor reveals several insights on the Co-NOMA scheme.
It can be noticed that each fading curve has its peak at different values of $\nu$. For instance, the best factor at $15$ dB additional fading is $0.5$, which infers that half of the time should be reserved for relaying. Contrary to that, $3$ dB fading has a clear global optimum at $\nu=1$, which implies that relaying is mostly ignored.
For all curves, we note that the local optimum shifts toward the right hand side, i.e., toward $\nu=1$, when decreasing the weak device additional fading.
Another observation from Fig.~\ref{res:yLogR_xTDMAfactor_FCP} is the existence of global and local optima in the low fading regime. In contrast, at the relatively high fading regime, i.e., beyond $6$ dB, only single stationary optimal points are observed. Such result is reasonable since the weak devices with high additional fading can hardly be served using the direct link. Hence the algorithm shifts toward activating more relay links, which results in lower values for $\nu$, i.e., more time allocated toward relaying. In the low fading regime, however, it is more favorable to allocate all the time toward direct links. Yet, as soon as there is some time allocated toward relaying (i.e., $\nu\leq 0.9$), we again observe the influence of time-split factor on the log-rate, i.e., the throughput and fairness trade-off. Such results emphasize the need to carefully tune all optimization parameters jointly, in order to achieve the best log-rate.
Thus, Fig.~\ref{res:yLogR_xTDMAfactor_FCP} shows that there is a trade-off between the time-split (direct serving vs. relay activation) and the weak device channel quality. Such observation highlights how the proposed scheme is able to flexibly adjust to the network deployment case and achieves optimal values if the time-split factor is taken into account.
In Fig.~\ref{res:yLogR_xTDMAfactor_PDP}, we show the same comparisons for the DRM Co-NOMA scheme. While the overall log-rate decreases in such cases, the major difference to Fig.~\ref{res:yLogR_xTDMAfactor_FCP} is the behavior of $6$dB additional fading. While for CRM, the global optimum is at $\nu=1$, in DRM the best log-rate is achieved with $\nu=0.7$. Hence, relaying is of greater importance for a distributed scheme, as it can better cope with the interference problems of DRM. In other words, through serving more weak devices trough relaying, parts of the interference may be mitigated, which is particularly problematic for distributed schemes.
In summary, Fig.~\ref{res:yLogR_xTDMAfactor_FCP} and Fig.~\ref{res:yLogR_xTDMAfactor_PDP} underline the need to further account for the time-split factor when designing a Co-NOMA scheme, which needs to balance both direct link serving and relaying through strong devices.
\begin{figure}[t]
\centering
\includegraphics[width=.99\linewidth]{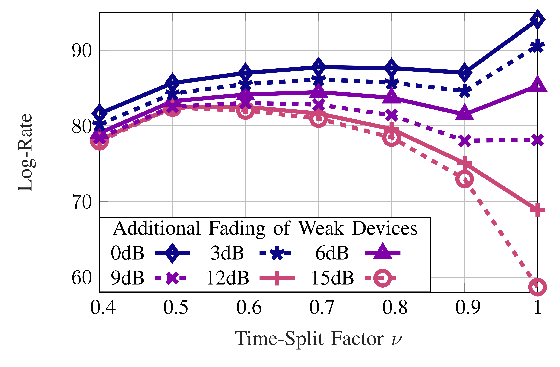}
\vspace*{-.35cm}
\caption{Log-rate over the time-split factor $\nu$ for the CRM Co-NOMA.}
\label{res:yLogR_xTDMAfactor_FCP}
\end{figure}
\begin{figure}[t]
\centering
\includegraphics[width=.99\linewidth]{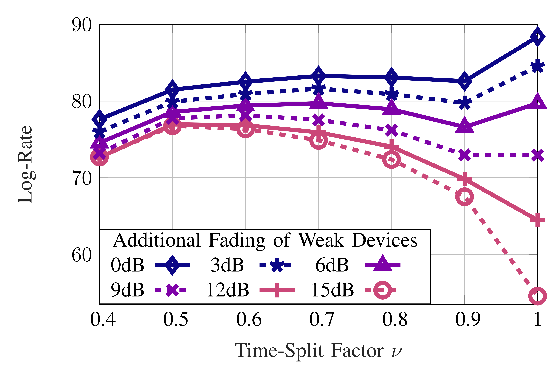}
\vspace*{-.35cm}
\caption{Log-rate over the time-split factor $\nu$ for the DRM Co-NOMA.}
\label{res:yLogR_xTDMAfactor_PDP}
\end{figure}
\begin{figure}[t]
\centering
\includegraphics[width=.99\linewidth]{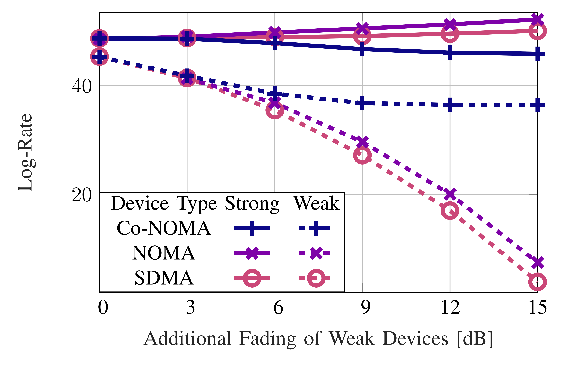}
\vspace*{-.35cm}
\caption{Log-rate over the weak device additional fading for the CRM Co-NOMA.}
\label{res:yLogR_xAddFading_fcp}
\end{figure}
\begin{figure}[t]
\centering
\includegraphics[width=.99\linewidth]{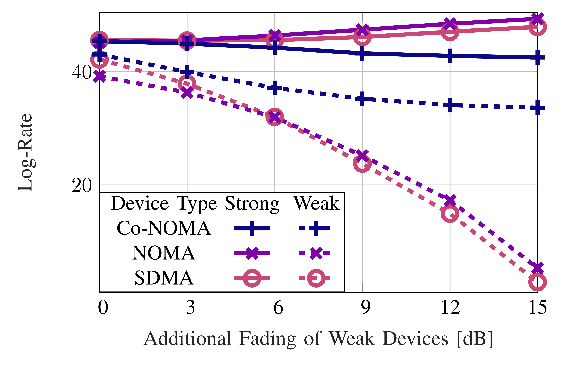}
\vspace*{-.35cm}
\caption{Log-rate over the weak device additional fading for the DRM Co-NOMA.}
\label{res:yLogR_xAddFading_pdp}
\end{figure}

\textbf{Log-rate per device type:} 
To gain further insights of the proposed Co-NOMA framework, we next strike a comparison of the log-rate of both the strong and weak devices over the weak device additional fading in Fig.~\ref{res:yLogR_xAddFading_fcp} (CRM) and Fig.~\ref{res:yLogR_xAddFading_pdp} (DRM).
In general, an observation is that the weak device log-rate decreases with decreasing the channel quality. This decrease is particularly severe for SDMA and NOMA, which can hardly tackle the changing channel conditions. We note that NOMA has measurable gains over SDMA, as the additional decoding capability at the strong device can mitigate some of the channel quality impairments. Co-NOMA, however, experiences minor decrease in weak device log-rate up to $9$ dB, after which the log-rate becomes almost constant.
Interestingly, for the strong devices, Fig.~\ref{res:yLogR_xAddFading_fcp} shows that the log-rates increase with the rising weak device additional fading for both SDMA and NOMA. That is, while the weak devices experience decreased quality-of-service, the remaining resources, which cannot empower the weak devices, are assigned to at least enhance the strong device service. Co-NOMA, however, behaves more tactfully, i.e., the log-rate of the strong devices decreases only slightly.
A noteworthy contribution of the proposed scheme is the characteristic that Co-NOMA sacrifices parts of the strong device log-rate to keep the weak device log-rates stabilized even at unfavorable channel conditions.
For instance, at $15$dB additional fading, there is only a minor log-rate loss comparing Co-NOMA's strong devices to SDMA or NOMA, while there is a major log-rate gain for the weak devices.

\textbf{Fairness index:} 
As a last remark on the impact of blockage, we show how the network fairness is affected as a function of the weak device additional fading. In Tab.~\ref{tb:jain}, we compare Co-NOMA, NOMA, and SDMA for the DRM, each in terms of Jain's fairness index versus the weak device additional fading. Jain's fairness index hereby depicts a measure for the network fairness, i.e., ${\left(\sum_{k\in\mathcal{K}} r_k\right)^2}/{K\cdot \sum_{k\in\mathcal{K}} r_k^2}$ \cite{jain1984quantitative}. When there is no additional fading, the indices of all schemes are relatively close to each other. 
When the additional fading increases, Jain's fairness index for NOMA and SDMA significantly degrades. However, Co-NOMA shows the opposite behavior, i.e., the fairness index grows with additional fading. In relation to Fig.~\ref{res:yLogR_xAddFading_all}, which illustrates the log-rate versus additional fading, Co-NOMA has a slightly reduced but stable log-rate, yet is able to increase the fairness among devices. 
The results of Tab.~\ref{tb:jain} clearly show the equitable performance of Co-NOMA, especially when the channel conditions become challenging. In fact, Co-NOMA exhibits promising and effective performance in situations where the weak device channels are weak, while the other schemes, i.e., NOMA and SDMA, fail to operate effectively.
\begin{table}[t]
\renewcommand{\arraystretch}{1.1}
\centering
\begin{tabular}{c  c  c  c  c  c  c}
	\hline 
	Fading [dB] \hspace*{-.2cm}& $0$ \hspace*{-.2cm}& $3$ \hspace*{-.2cm}& $6$ \hspace*{-.2cm}& $9$ \hspace*{-.2cm}& $12$ \hspace*{-.2cm}& $15$ \\\hline
	SDMA \hspace*{-.2cm}& 0.6553 \hspace*{-.2cm}& 0.6419 \hspace*{-.2cm}& 0.6075 \hspace*{-.2cm}& 0.5530 \hspace*{-.2cm}& 0.5096 \hspace*{-.2cm}& 0.4636\\
	NOMA \hspace*{-.2cm}& 0.6969 \hspace*{-.2cm}& 0.6676 \hspace*{-.2cm}& 0.6243 \hspace*{-.2cm}& 0.5658 \hspace*{-.2cm}& 0.5106 \hspace*{-.2cm}& 0.4601\\
	Co-NOMA \hspace*{-.2cm}& 0.6780 \hspace*{-.2cm}& 0.6882 \hspace*{-.2cm}& 0.6949 \hspace*{-.2cm}& 0.7037 \hspace*{-.2cm}& 0.7021 \hspace*{-.2cm}& 0.7077\\\hline
\end{tabular}
\vspace*{-.15cm}
\caption{Jain's fairness index over the weak device additional fading for the DRM schemes.}
\label{tb:jain}
\end{table}%

The proposed Co-NOMA scheme outperforms the reference SDMA and NOMA schemes both in terms of weak device log-rate and network fairness perspectives. Such results, in fact, emphasize the equitable service potential of the proposed optimized Co-NOMA scheme, and its applicability to several XR use-cases, e.g., telepresence experiences, model visualization, and extended teaching methodologies.

\subsection{Residual UAV Power}
\begin{figure}[t]
\centering
\includegraphics[width=.99\linewidth]{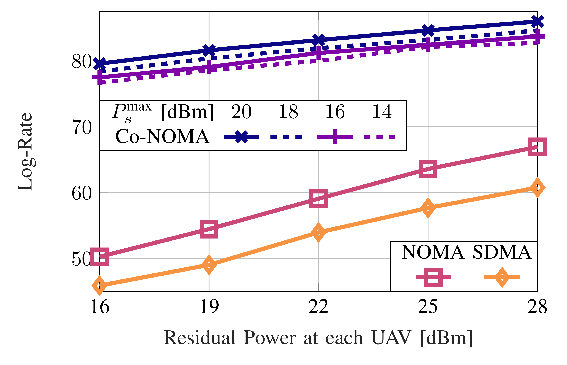}
\vspace*{-.35cm}
\caption{Log-rate over the residual transmit and computation power at the UAVs ($P_e^\text{max}-\Phi_e$) for the CRM schemes.}
\label{res:yLogR_xUAVpwr_fcp}
\end{figure}
The EC-enabled UAVs have strict power consumptions, i.e., the resource management accounts for transmission, computation, and operational power. While the operational power $\Phi_e$ is fixed, the residual power available at each UAV, denoted by $P_e^\text{max}-\Phi_e$, needs to be flexibly adjusted to the network computation and communication needs.

Hence, in the following set of simulations, we set the weak device additional fading to $12$ dB, $\nu=0.6$, and vary two parameters, namely, $P_e^\text{max}-\Phi_e$, i.e., the residual power at each UAV, and $P_s^\text{max}$, i.e., the strong device transmit power. Fig.~\ref{res:yLogR_xUAVpwr_fcp} shows the log-rate of the CRM schemes as a function of the residual power. First, a general observation is the increase of log-rate with increasing power values, which is reasonable as the strict UAV power limitations pose a major bottleneck in the system model. While SDMA and NOMA depict increasing log-rate behavior as the residual power increases, Co-NOMA already outperforms both reference schemes on all levels. Specifically, in low-power regions, i.e., the most strict residual power constraint for the UAVs, Co-NOMA has the largest gain over SDMA and NOMA.
Further, as seen in Fig.~\ref{res:yLogR_xUAVpwr_fcp}, the strong device maximum transmit power $P_s^\text{max}$ has observable impacts on the performance of Co-NOMA. While a value of $20$ dBm achieves the highest performance, a $14$dBm value achieves the lowest log-rate among Co-NOMA schemes, which is expected as relaying plays an important role in the proposed architecture. Yet, even under the strictest D2D power constraint, Co-NOMA significantly outperforms the baselines. 


The above results related to the log-rate performance versus different power constraints underline the appreciable improvements of the proposed schemes as compared to state-of-the-art. Particularly in strict regions, i.e., under low residual and strong device powers, which are most realistic as the usual XR device is battery-limited, Co-NOMA achieves optimized performance levels.

\subsection{Computation Capacity}
\begin{figure}[t]
\centering
\includegraphics[width=.99\linewidth]{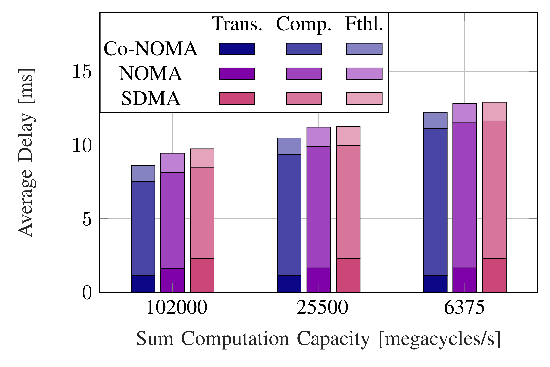}
\vspace*{-.35cm}
\caption{Average delay in ms over the sum computation capacity for the CRM Co-NOMA and CRM SDMA comparing transmission (trans.), computation (comp.), and fronthaul (fthl.) delay.}
\label{res:yAvgDelay_xCompCap_fcp}
\end{figure}
While the log-rate demonstrates a trade-off between throughput and network fairness, recall that the proposed method incorporates further aspects of the joint communication and computations, which are of particular interest for enabling delay-sensitive XR use-cases, e.g., immersive gaming and telepresence. Hence, in the following set of simulations, we let the additional fading be $9$ dB, $R_b^\text{max}=50$ Mbps, $F_k = 10^6$ cycles, $D_k = 10^5$ bits, and vary the sum computation capacity to illustrate the delay performance of the proposed algorithm. That is, we initially let $f_0^\text{max}=10\cdot10^{10}$ cycles/s and $f_e^\text{max}=2\cdot10^{9}$ cycles/s and proportionally reduce such values for each simulation point (e.g., the x-axis in Fig.~\ref{res:yAvgDelay_xCompCap_fcp}).

To this end, Fig.~\ref{res:yAvgDelay_xCompCap_fcp} depicts the average delay per device in ms versus the sum computation capacity. Fig.~\ref{res:yAvgDelay_xCompCap_fcp}, in particular, shows the average delay split into transmission, computation, and fronthaul delay components for each CRM scheme.
As expected, Fig.~\ref{res:yAvgDelay_xCompCap_fcp} shows that the average delay of all schemes increases while decreasing the sum computation capacity, which is reasonable as the computation delay increases for the average device.
However, for all values of the considered sum computation capacity, Co-NOMA has the lowest average delay as compared to NOMA and SDMA. Such behavior is even more visible when comparing only the transmission delays. Co-NOMA not only reduces the transmission delay, but also provides a lower average delay, while accounting both for computation and for fronthauling aspects.

\begin{figure}[t]
\centering
\includegraphics[width=.99\linewidth]{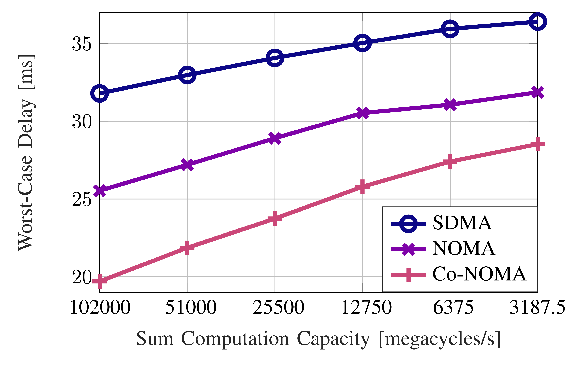}
\vspace*{-.35cm}
\caption{Worst-case delay in ms over the sum computation capacity for the CRM schemes.}
\label{res:yWCD_xUAVpwr_fcp}
\end{figure}
Fig.~\ref{res:yWCD_xUAVpwr_fcp} now shows the worst-case delay for the same simulation set, i.e., the device with the worst delay performance.
As before, one can see that the worst-case delay increases with a decrease in the sum computation capacity. That is, the computing capability of the CC/ECs has a significant impact on the overall delay.
Once again, the Co-NOMA scheme clearly outperforms the benchmarking schemes in terms of worst-case delay, which underlines its applicability for XR delay sensitive applications. Together with the insights on the individual delay composition, it comes clear that our optimized Co-NOMA is a promising technology with several advantages in terms of delay, throughput, and fairness.

\subsection{Scalability in Dense Networks}
\begin{figure}[t]
\centering
\includegraphics[width=.99\linewidth]{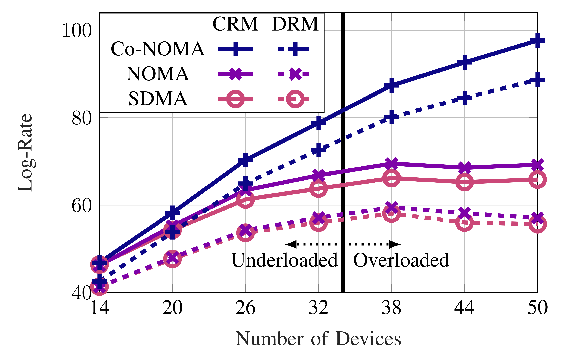}
\vspace*{-.35cm}
\caption{Log-rate over the number of devices comparing all schemes in underloaded and overloaded networks.}
\label{res:yLogR_xNumU}
\end{figure}
To assess the performance of the considered scheme in dense networks, we now consider different network sizes. Specifically, we vary the number of devices sharing the same time and frequency resource from $14$ to $50$, by considering a network consisting of $3$ BSs, each equipped with $6$ antennas, and $4$ UAVs, each equipped with with $4$ antennas. That is, we consider underloaded and overloaded network deployments. In underloaded networks, the number of total available transmit antennas is higher than the number of devices, and vice versa. By linearly searching over $\nu$, by using $T_k = 100$ ms, and by setting the additional fading to $9$ dB, Fig.~\ref{res:yLogR_xNumU} plots the log-rate versus the number of devices for all considered schemes. The figure shows that the log-rate increases jointly with the number of devices in case of underloaded networks for all schemes. This is reasonable as more devices bring better opportunities to increase the throughput and the fairness up to some point. As can be seen for SDMA, the log-rate experiences a saturation beyond a certain value in terms of the device numbers, i.e., in overloaded networks. This is particularly the case since SDMA depends on the number of antennas needed to orthogonalize the channels for efficient transmission. As the number of devices overtakes the number of antennas, however, the log-rate can no longer increase. A similar behavior is visible for NOMA, which, however, performs better than SDMA in overloaded networks due to the more sophisticated interference management. In contrast, Co-NOMA can, through the efficient utilization of relay links, keep increasing the log-rate so as to cope with the fairness issue of dense networks. Interestingly, DRM Co-NOMA even outperforms the centralized versions of both SDMA and NOMA when the number of devices exceeds $26$, which emphasizes the advantages of the proposed decentralized method in future dense XR networks.

\begin{figure}[t]
\centering
\includegraphics[width=.99\linewidth]{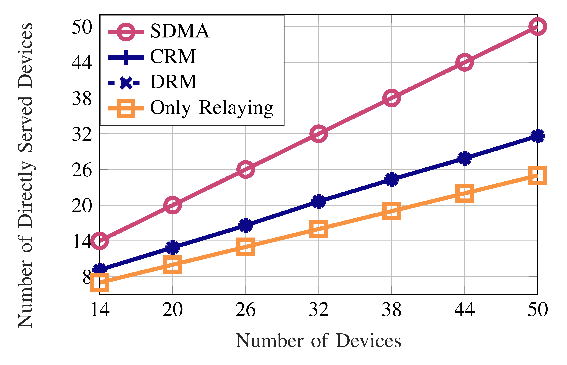}
\vspace*{-.35cm}
\caption{The number of direct served weak devices over the number of devices.}
\label{res:yNumDirect_xNumU}
\end{figure}
Such results are further verified in Fig.~\ref{res:yNumDirect_xNumU}, which shows the number of directly served devices (i.e., without relaying) versus the total number of devices. As the number of devices increases, the Co-NOMA scheme makes better use of the relay links. Fig.~\ref{res:yNumDirect_xNumU} also shows that both CRM and DRM keep a relatively similar slope as the \emph{Only Relaying} case, whereas the distance to the SDMA (only direct links) grows as the number of devices increases. That is, the device relaying aspect of Co-NOMA becomes especially relevant for dense networks.

To summarize, the results in this subsection highlight the performance gains of Co-NOMA, especially when the network becomes dense. A special result is that DRM Co-NOMA outperforms the centralized reference schemes when the number of devices is large. Therefore, with a lower runtime and complexity, higher values of the log-rate are attained by DRM Co-NOMA, which provides the best throughput and fairness trade-off among all simulated schemes.

\subsection{Runtime}
\begin{table}[t]
\renewcommand{\arraystretch}{1.0}
\centering
\begin{tabular}{c c c c c c c c}
	\hline
	\hspace*{-.2cm}Number of devices \hspace*{-.25cm}& 	$14$ \hspace*{-.25cm}& $20$ \hspace*{-.25cm}& $26$ \hspace*{-.25cm}& $32$ \hspace*{-.25cm}& $38$ \hspace*{-.25cm}& $44$ \hspace*{-.25cm}& $50$\\\hline
	\hspace*{-.2cm}CRM SDMA \hspace*{-.25cm}& 			$25.8$ \hspace*{-.25cm}& $39.2$ \hspace*{-.25cm}& $56.3$ \hspace*{-.25cm}& $84.2$ \hspace*{-.25cm}& $116.7$ \hspace*{-.25cm}& 	$143.7$ \hspace*{-.25cm}& 	$173.5$\\
	\hspace*{-.2cm}CRM NOMA \hspace*{-.25cm}&  			$28.3$ \hspace*{-.25cm}& $44.1$ \hspace*{-.25cm}& $72.8$ \hspace*{-.25cm}& $113.4$ \hspace*{-.25cm}& $161.8$ \hspace*{-.25cm}& $218.7$ \hspace*{-.25cm}& $282.2$\\
	\hspace*{-.2cm}CRM Co-NOMA \hspace*{-.25cm}& 		$40.4$ \hspace*{-.25cm}& $64.9$ \hspace*{-.25cm}& $103.1$ \hspace*{-.25cm}& $150.6$ \hspace*{-.25cm}& $208.6$ \hspace*{-.25cm}& $281.9$ \hspace*{-.25cm}& $375.5$\\
	\hspace*{-.2cm}DRM SDMA \hspace*{-.25cm}& 			$13.5$ \hspace*{-.25cm}& $17.3$ \hspace*{-.25cm}& $22.0$ \hspace*{-.25cm}& $26.0$ \hspace*{-.25cm}& $31.3$ \hspace*{-.25cm}& $35.6$ \hspace*{-.25cm}& $39.1$\\
	\hspace*{-.2cm}DRM NOMA \hspace*{-.25cm}&  			$14.1$ \hspace*{-.25cm}& $17.6$ \hspace*{-.25cm}& $21.8$ \hspace*{-.25cm}& $27.0$ \hspace*{-.25cm}& $33.0$ \hspace*{-.25cm}& $38.1$ \hspace*{-.25cm}& $45.3$\\
	\hspace*{-.2cm}DRM Co-NOMA \hspace*{-.25cm}& 		$20.7$ \hspace*{-.25cm}& $26.8$ \hspace*{-.25cm}& $33.1$ \hspace*{-.25cm}& $39.1$ \hspace*{-.25cm}& $46.4$ \hspace*{-.25cm}& $53.9$ \hspace*{-.25cm}& $63.0$\\\hline
\end{tabular}
\vspace*{-.15cm}
\caption{Runtime values in seconds over $K$ for all schemes.}
\label{tb:runtime}
\end{table}%
In Section~\ref{ssec:compcomp}, the computational complexity of the DRM Co-NOMA is shown to be significantly lower than the CRM's complexity. To further observe the computation characteristics of the proposed methods, we next evaluate the runtimes of all considered schemes as a function of the number of devices in Tab.~\ref{tb:runtime}, for the same parameter set as in the last subsection.
The first and most obvious observation is the joint increase of runtime as the number of devices increases.
Comparing only the centralized schemes, Co-NOMA has the highest runtime and SDMA the lowest, which highly correlates to the optimization algorithmic complexity. Thus, CRM Co-NOMA can provide the best log-rate performance, which comes at the cost of longer execution time.
Among DRM schemes, Co-NOMA has a factor of around $1.5$ times the runtime of SDMA for all number of devices. Another important observation is the comparison of DRM Co-NOMA and CRM SDMA or NOMA. That is, DRM Co-NOMA provides significantly faster runtimes than CRM SDMA or NOMA at $K=50$, and also outperforms both schemes in terms of log-rate, as illustrated earlier in Fig.~\ref{res:yLogR_xNumU}. This noticeable behavior of DRM Co-NOMA illustrates its numerical superiority in terms of both performance and complexity, thereby making it a compelling solution for future XR systems.
To sum up, on the one hand, while CRM Co-NOMA has the highest expected runtime, it also provides the best log-rate performance. On the other hand, the more practical, more scalable, DRM Co-NOMA achieves better runtimes at the cost of negligible log-rate performance loss. Such results highlight the potential applicability of the proposed decentralized Co-NOMA scheme in future wireless networks for promoting both high throughput and high fairness values.

\subsection{Convergence}
\begin{figure}[t]
\centering
\includegraphics[width=.99\linewidth]{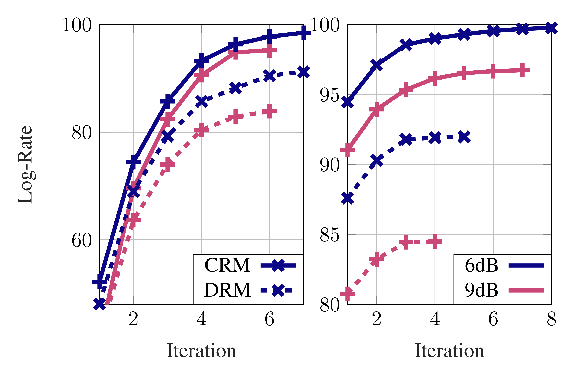}
\vspace*{-.35cm}
\caption{Convergence behavior in terms of log-rate over the iteration number for the Co-NOMA scheme. The convergence behavior of the first loop (left), i.e., under $\ell_0$-norm relaxation, and the second loop (right), i.e., with fixed link selection and clustering, is shown.}
\label{res:yLogR_xIteration}
\end{figure}
Lastly, we numerically validate the convergence of the proposed methods in Fig.~\ref{res:yLogR_xIteration}. The left hand subfigure of Fig.~\ref{res:yLogR_xIteration} shows the log-rate value versus the number of each iterations using both Algorithm~\ref{alg:drm} and Algorithm~\ref{alg:crm}. While CRM converges to a larger log-rate than DRM, both schemes converge in few iterations. At this stage, reasonable link selection and clustering sets are found, we then run both Algorithm~\ref{alg:crm} and Algorithm~\ref{alg:drm} with reduced optimization variables and show the convergence on the right hand side subfigure of Fig.~\ref{res:yLogR_xIteration}. At $6$ dB additional weak device fading, CRM converges within $8$ iterations, while DRM converges within $5$ iterations. This is reasonable as the CRM searches over a larger optimization space due to its centralized perspective.
Further, Fig.~\ref{res:yLogR_xIteration} shows that the algorithms converges even faster when the weak device additional fading is higher, e.g., $9$dB. While the channel condition of the weak devices worsens, the algorithmic link selections are more in favor of the relay links. The opportunities for optimization are also becoming less with worse channel conditions, which is why such result is expected.
However, both schemes converge within few iterations and particularly the DRM provides reasonable and practically feasible resource allocations for optimized log-rate under practical constraints.

\section{Conclusion}\label{sec:con}
The connectivity prospects of 6G communication networks empower the applications of XR, especially those relevant to the medical, educational, engineering, and social contexts. Such performance-sensitive applications necessitate powerful joint communication and computation, connectivity, and delay measures, as well as digitally equity considerations.
To this end, this paper considers a hybrid CC/MEC network consisting of multi-antenna BSs and UAVs, which collaboratively serve a number of XR devices. The CC manages the core-network functions and connects to the BSs via fronthaul links. The UAVs, on the other hand, perform communication, computation, and resource management at the network edge as MEC decentralized platforms. Further, to jointly enhance the system fairness and cope with the adverse channel conditions, we employ a Co-NOMA framework, where the strongly connected XR devices may aid the communication toward the weakly connected XR devices using device-to-device relaying. The paper then formulates a sum logarithmic-rate maximization problem so as to determine the joint computation allocation, beamforming vectors, power allocations, link-selection, and clustering as a means to strike a trade-off between the system throughput and fairness. With the help of techniques from optimization theory, e.g., $\ell_0$-norm relaxation, SCA, and FP, a decentralized algorithm for managing the network resources is proposed. The numerical simulations ot the paper then highlight the appreciable improvement of Co-NOMA over the reference schemes in terms of log-rate, delay performance, and scalability. The proposed optimized decentralized DRM is particularly shown to perform close to the centralized scheme, to outperform the benchmarks in various setups, and to exhibit a valuable scalabilty in terms of complexity and runtime in dense networks, which makes it a valuable candidate for future decentralized XR systems. 


%
%

\appendices

\section{Proof of Lemma \ref{lma1}}\label{app1}
Recall that equation \eqref{eq:rnsc3} is of bilinear nature due to the terms $z_{w}^{(d)} r_{w}^{(d)}$ and $z_{w}^{(r)} r_{w}^\text{aux}$.
An equivalent mathematical formulation of the former term $z_{w}^{(d)} r_{w}^{(d)}$ into a formulation comprising two quadratic functions is given as
\begin{equation}\label{eq:lma1}
z_{w}^{(d)} r_{w}^{(d)} = \frac{1}{4}\left((z_{w}^{(d)}+_{w}^{(d)})^2-(z_{w}^{(d)}-_{w}^{(d)})^2\right).
\end{equation}
A similar reformulation can also be applied to $z_{w}^{(r)} r_{w}^\text{aux}$.

%

Using \eqref{eq:lma1}, we define $\chi_w(\bm{z},\bm{r})$ and can rewrite equation \eqref{eq:rnsc3} into
\begin{align}
\chi_w(\bm{z},\bm{r}) =& 4 r_w - \big((z_{w}^{(d)} + r_{w}^{(d)})^2 - (z_{w}^{(d)} - r_{w}^{(d)})^2\big) \nonumber\\
&- \big( (z_{w}^{(r)} + r_{w}^\text{aux})^2 - (z_{w}^{(r)} - r_{w}^\text{aux})^2 \big) \leq 0.\label{eq:rnsc4}
\end{align}
Some parts of \eqref{eq:rnsc4} remain concave, and so we next apply the SCA framework in order to obtain a convex formulation. Hence, we first identify the reason of non-convexity in \eqref{eq:rnsc4}, which stems from $-(z_{w}^{(d)} + r_{w}^{(d)})^2$ and $-(z_{w}^{(r)} + r_{w}^\text{aux})^2$. To tackle this non-convexity, we consider the following remark.
\begin{remark}
The first order Taylor extension of a function $\Psi(\boldsymbol{\tau})$, approximates $\Psi(\boldsymbol{\tau})$ around the operating point $\boldsymbol{\tilde{\tau}}$
\begin{equation}
\Psi(\boldsymbol{\tau}) \approx \Psi(\boldsymbol{\tilde{\tau}}) + \nabla_{\boldsymbol{\tau}} \Psi(\boldsymbol{\tilde{\tau}})^T (\boldsymbol{\tau}-\boldsymbol{\tilde{\tau}}),
\end{equation}
which is a lower bound for convex functions and an upper bound for concave functions.
\end{remark}%
Now, applying the first order Taylor extension to $-(z_{w}^{(d)} + r_{w}^{(d)})^2$ and $-(z_{w}^{(r)} + r_{w}^\text{aux})^2$ in \eqref{eq:rnsc4}, we obtain
\begin{align}
&\tilde{\chi}_w(\bm{z},\bm{r},\bm{\tilde{z}},\bm{\tilde{r}}) = 4 r_w + ({z}_{w}^{(d)} - {r}_{w}^{(d)})^2 - (\tilde{z}_{w}^{(d)} + \tilde{r}_{w}^{(d)})^2 \label{eq:rnsc42}\\
&- 2 (\tilde{z}_{w}^{(d)} + \tilde{r}_{w}^{(d)})\big[(z_w^{(d)}-\tilde{z}_w^{(d)})+({r}_{w}^{(d)}-\tilde{r}_{w}^{(d)}) \big]\nonumber\\
&+ (z_{w}^{(r)} - r_{w}^\text{aux})^2 - (\tilde{z}_{w}^{(r)} + \tilde{r}_{w}^\text{aux})^2 \nonumber\\
&- 2 (\tilde{z}_{w}^{(r)} + \tilde{r}_{w}^\text{aux}) \big[({z}_{w}^{(r)}-\tilde{z}_{w}^{(r)})+({r}_{w}^\text{aux}-\tilde{r}_{w}^\text{aux})\big]\leq 0,\nonumber
\end{align}
where $\bm{\tilde{z}}$ and $\bm{\tilde{r}}$ are optimal solutions reached from the previous iteration.
The following remark and steps of this proof are drawn in analogy to the sequential optimization tool \cite[Proposition 3]{7862919} and \cite{marks1978general}.
\begin{remark}
Applying the SCA technique to the constraints of a maximization problem with differentiable objective function and compact feasible set, such as \eqref{eq:Opt1}, yields a sequence of optimization problems with monotonically increasing optimal values. This sequence converges to a finite limit, which satisfies the KKT conditions of the original problem. For further details, the authors refer to work \cite{7862919}.
\end{remark}
In this case, suppose that the sequence of optimization problems after applying lemma~\ref{lma1} differ to the problem before applying such lemma in one constraint only, namely ${\chi}_w(\bm{z},\bm{r})\leq 0$ (in the original problem) and $\tilde{\chi}_w(\bm{z},\bm{r},\bm{\tilde{z}},\bm{\tilde{r}})\leq 0$ (in each reformulated problem). Due to the nature of the first order Taylor extension, it can be verified that $\tilde{\chi}_w(\bm{z},\bm{r},\bm{\tilde{z}},\bm{\tilde{r}}) \geq {\chi}_w(\bm{z},\bm{r})$. Also, using equations \eqref{eq:rnsc4} and \eqref{eq:rnsc42}, we see that $\tilde{\chi}_w(\bm{\tilde{z}},\bm{\tilde{r}},\bm{\tilde{z}},\bm{\tilde{r}}) = {\chi}_w(\bm{\tilde{z}},\bm{\tilde{r}})$ holds. At last, one can validate that $\nabla\tilde{\chi}_w(\bm{\tilde{z}},\bm{\tilde{r}},\bm{\tilde{z}},\bm{\tilde{r}}) = \nabla{\chi}_w(\bm{\tilde{z}},\bm{\tilde{r}})$, which yields the following:
\begin{align}
\nabla\tilde{\chi}_w(\bm{\tilde{z}},\bm{\tilde{r}},\bm{\tilde{z}},\bm{\tilde{r}}) = \nabla{\chi}_w(\bm{\tilde{z}},\bm{\tilde{r}}) = -4 \begin{bmatrix}
\tilde{r}_w^{(d)} \\
\tilde{z}_w^{(d)} \\
\tilde{r}_w^{\text{aux}} \\
\tilde{z}_w^{(r)}
\end{bmatrix}.
\end{align}
This completes the proof.

\section{Proof of Lemma \ref{lma3}}\label{app2}
Consider the following general remark on the quadratic transform.
\begin{remark}
The quadratic transform decouples the numerator and denominator of a fractional function $\boldsymbol{\Psi}_n^H(\boldsymbol{\tau})\boldsymbol{\Psi}_d^{-1}(\boldsymbol{\tau})\boldsymbol{\Psi}_n(\boldsymbol{\tau})$ by using an auxiliary variable $\bm{\alpha}$ as
\begin{equation}
2 \text{Re}\left\{ \boldsymbol{\alpha}^H \boldsymbol{\Psi}_n(\boldsymbol{\tau}) \right\} - \boldsymbol{\alpha}^H \boldsymbol{\Psi}_d(\boldsymbol{\tau})\boldsymbol{\alpha}.\label{eq:qt2}
\end{equation}
These reformulations introduce the auxiliary variable $\boldsymbol{\alpha}$, which can be updated in an outer loop.
The optimal auxiliary variables are computed by setting the partial derivative of \eqref{eq:qt2} with respect to $\boldsymbol{\alpha}$ to zero and solving for $\boldsymbol{\alpha}$. For further details on FP, the authors refer to work \cite{FP1}.
\end{remark}
Recall the first constraint in \eqref{eq:SINR_all}, which is given as
\begin{equation}
\gamma_{w}^{(r)} \leq 
|\bm{h}_{d_w}^{H} {\bm{q}_{w}^{(r)}}|^2/I_{d_w}(\mathcal{K}^\text{str},\mathcal{K}^\text{weak}\backslash\{w\}),\label{eq:gammawr}
\end{equation}
with
\begin{align}
I_{d_w}&(\mathcal{K}^\text{str},\mathcal{K}^\text{weak}\backslash\{w\}) = \sigma^2 + \sum_{i'\in\mathcal{K}^\text{str}}\big|\bm{h}_i^H\bm{q}_{i'}^{(d)}\big|^2 \nonumber\\
&\hspace*{1.8cm}+ \sum_{j'\in\mathcal{K}^\text{weak}\backslash\{w\}}\big|\bm{h}_i^H\big(\bm{q}_{j'}^{(d)}+\bm{q}_{j'}^{(r)}\big)\big|^2, \label{eq:I_dw}
\end{align}
To apply the quadratic transform, i.e., \eqref{eq:qt2}, we identify $\boldsymbol{\chi}_d(\bm{x})=I_{d_w}(\mathcal{K}^\text{str},\mathcal{K}^\text{weak}\backslash\{w\})$. Further, given $|\bm{h}_{d_w}^{H} {\bm{q}_{w}^{(r)}}|^2 = \boldsymbol{\chi}_n^H(\bm{x})\boldsymbol{\chi}_n(\bm{x})$, it can be seen that $ \boldsymbol{\chi}_n(\bm{x})=\bm{h}_{d_w}^{H} {\bm{q}_{w}^{(r)}}$. Defining the auxiliary variable $a_w^{(r)}$ and applying the quadratic transform to \eqref{eq:gammawr}, equation \eqref{eq:gsw1} follows directly. The task now remains in determining the optimal auxiliary variable $a_w^{(r)}$. Hence we first take the partial derivative of \eqref{eq:gsw1} with respect to $a_w^{(r)}$
\begin{align}
0&=\frac{\partial}{\partial a_w^{(r)}}\Big(\gamma_{w}^{(r)} - 2 \text{Re}\big\{ (a_{w}^{(r)})^H \bm{h}_{d_w}^{H} \bm{q}_{w}^{(r)} \big\} \nonumber\\
&\hspace*{2cm}+ |a_{w}^{(r)}|^2 I_{d_w}(\mathcal{K}^\text{str},\mathcal{K}^\text{weak}\backslash\{w\})\Big),\label{eq:gsw1__} \\
&=2\text{Re}\big\{\bm{h}_{d_w}^{H} \bm{q}_{w}^{(r)}\big\} + 2(a_{w}^{(r)})^* I_{d_w}(\mathcal{K}^\text{str},\mathcal{K}^\text{weak}\backslash\{w\}).
\end{align}
The optimal auxiliary variable $(a_w^{(r)})^*$ can then be written as:
\begin{align}
(a_{w}^{(r)})^* &= \text{Re}\big\{\bm{h}_{d_w}^{H} \bm{\tilde{q}}_{w}^{(r)} \big\} /  \tilde{I}_{d_w}(\mathcal{K}^\text{str},\mathcal{K}^\text{weak}\backslash\{w\}). \label{eq:amn1_}
\end{align}
Similar derivations are readily applied to the constraints in \eqref{eq:SINR_all}, which completes the proof.

\bibliographystyle{IEEEtran}
\bibliography{bibliography}
\end{document}